\documentclass[11pt]{article}
\usepackage{amssymb}
\usepackage{amsfonts}
\usepackage{amsmath}
\usepackage{amsthm}
\usepackage{latexsym}

\usepackage{mathtools}

\usepackage{xspace} 

\usepackage[pagebackref=true]{hyperref}
\hypersetup{
    unicode=false,          
    colorlinks=true,        
    linkcolor=blue,          
    citecolor=blue,        
    filecolor=magenta,      
    urlcolor=cyan           
}

\usepackage{hyperref,cleveref}
\usepackage{amsmath}
\usepackage{amsthm}
\usepackage{amsfonts}
\usepackage{fullpage,appendix}

\usepackage{enumitem}

\usepackage[ruled]{algorithm}
\usepackage{algpseudocode}
\usepackage{algorithmicx}
\floatname{algorithm}{Protocol}

\usepackage{dsfont}
\usepackage{color}
\usepackage{tikz}

\usepackage{float}

\newtheorem{fact}{Fact}[section]

\newtheorem*{rep@theorem}{\rep@title}
\newcommand{\newreptheorem}[2]{%
\newenvironment{rep#1}[1]{%
 \def\rep@title{#2 \ref{##1}}%
 \begin{rep@theorem}}%
 {\end{rep@theorem}}}
\makeatother

\newtheorem{theorem}[fact]{Theorem}
\newreptheorem{theorem}{Theorem}

\newreptheorem{lemma}{Lemma}

\usetikzlibrary{decorations.pathreplacing}

\newcommand{\ignore}[1]{}
\definecolor{corlinks}{RGB}{64,128,128}
\definecolor{cormenu}{RGB}{0,37,94}
\definecolor{corurl}{RGB}{0,46,91}
\definecolor{darkgreen}{rgb}{0,0.5,0}

\newcommand{\1}{\mathds{1}}

\newcommand{\err}{\mathsf{err}}

\newcommand{\agr}{\mathsf{agr}}

\newcommand{\Cov}{{\rm Cov}}

\DeclareMathOperator*{\argmax}{arg\,max}

\DeclareMathOperator*{\Var}{\mathsf{Var}}

\newcommand{\Sign}{\mathsf{Sign}}

\newcommand{\calf}{\mathcal{F}}

\newcommand{\floor}[1]{\left \lfloor #1 \right \rfloor}
\newcommand{\zo}{\{0, 1\}}

\newcommand{\Ex}{\mathbb E}

\newcommand{\Maj}{\mathsf{Maj}}

\newcommand{\eps}{\epsilon}

\newtheorem{definition}[fact]{Definition}
\newtheorem{defn}[fact]{Definition}

\newtheorem{lem}[fact]{Lemma}
\newtheorem{corollary}[fact]{Corollary}

\newtheorem{proposition}[fact]{Proposition}

\newtheorem*{note}{Note}

\newcommand{\CC}{\mathsf{CC}} 
\newcommand{\CCU}{\mathsf{PubCCU}} 
\newcommand{\owCC}{\mathsf{owCC}} 
\newcommand{\owCCU}{\mathsf{owPubCCU}} 

\newcommand{\owIsrCCU}{\mathsf{owIsrCCU}} 

\newcommand{\owIsrCC}{\mathsf{owIsrCC}} 

\newcommand{\GIP}{\mathsf{GIP}} 

\newcommand{\ISR}{\mathsf{ISR}} 

\newcommand{\SubsetMaj}{\mathsf{SubsetMaj}} 

\newcommand{\SubsetParity}{\mathsf{SubsetParity}} 

\newcommand{\IIC}{\textsc{IIC}}

\newcommand{\HD}{\mathsf{HD}}

\newcommand{\PrivCC}{\mathsf{PrivCC}} 
\newcommand{\owPrivCC}{\mathsf{owPrivCC}} 

\newcommand{\PrivCCU}{\mathsf{PrivCCU}} 

\newcommand{\owPrivCCU}{\mathsf{owPrivCCU}} 

\newcommand{\Inote}[1]{}
\newcommand{\Pnote}[1]{}
\newcommand{\Mnote}[1]{}

\newif\ifnotes\notestrue

\usepackage{empheq}

\title{The Power of Shared Randomness in Uncertain Communication} \author{
Badih Ghazi \thanks{MIT, CSAIL, Cambridge MA 02139. {\tt badih@mit.edu}.}
\and
Madhu Sudan\thanks{Harvard John A. Paulson School of Engineering and Applied Sciences. {\tt madhu@cs.harvard.edu}.}
}
\date{\today}

\begin{document}
\maketitle

\begin{abstract}
In a recent work (Ghazi et al., SODA 2016), the authors with Komargodski and Kothari initiated the study of \emph{communication with contextual uncertainty}, a setup aiming to understand how efficient communication is possible when the communicating parties imperfectly share a huge context. In this setting, Alice is given a function $f$ and an input string $x$, and Bob is given a function $g$ and an input string $y$. 
The pair $(x,y)$ comes from a known distribution $\mu$ and $f$ and $g$ are
guaranteed to be close under this distribution.
Alice and Bob wish to compute $g(x,y)$ with high probability. The lack of agreement between Alice and Bob on the function that is being computed captures the uncertainty in the context. The previous work showed that any problem with one-way communication complexity $k$ in the standard model (i.e., without uncertainty\footnote{in other words, under the promise that $f=g$.}) has \emph{public-coin} communication at most $O(k(1+I))$ bits in the uncertain case, where $I$ is the mutual information between $x$ and $y$. Moreover, a lower bound of $\Omega(\sqrt{I})$ bits on the public-coin uncertain communication was also shown.

However, an important question that was left open is related to the power that public randomness brings to uncertain communication. Can Alice and Bob achieve efficient communication amid uncertainty without using public randomness? And how powerful are public-coin protocols in overcoming uncertainty? Motivated by these two questions:
\begin{itemize}
\item We prove the first separation between private-coin uncertain communication and public-coin uncertain communication. Namely, we exhibit a function class for which the communication in the standard model and the public-coin uncertain communication are $O(1)$ while the private-coin uncertain communication is a growing function of $n$ (the length of the inputs). This lower bound (proved with respect to the uniform distribution) is in sharp contrast with the case of public-coin uncertain communication which was shown by the previous work to be within a constant factor from the certain communication. This lower bound also implies the first separation between public-coin uncertain communication and deterministic uncertain communication. Interestingly, we also show that if Alice and Bob \emph{imperfectly share} a sequence of random bits (a setup weaker than public randomness), then achieving a constant blow-up in communication is still possible.
\item We improve the lower-bound of the previous work on public-coin uncertain communication. Namely, we exhibit a function class and a distribution (with mutual information $I \approx n$) for which the one-way certain communication is $k$ bits but the one-way public-coin uncertain communication is at least $\Omega(\sqrt{k} \cdot \sqrt{I})$ bits.
\end{itemize}

Our proofs introduce new problems in the standard communication complexity model and
prove lower bounds for these problems. Both the problems and the lower bound
techniques may be of general interest.

\end{abstract}

\newpage

\tableofcontents

\thispagestyle{empty}
\newpage
\pagenumbering{arabic}

\section{Introduction}\label{sec:intro}

In many forms of communication (e.g., human, computer-to-computer), the communicating parties share some context (e.g, knowledge of a language, operating system, communication protocol, encoding/decoding mechanisms.). This context is usually a) \emph{huge} and b) \emph{imperfectly shared} among the parties. Nevertheless, in human communication, very efficient communication is usually possible. Can we come up with a mathematical analogue of this phenomenon of efficient communication based on a huge but imperfectly shared context? Motivated by this general question, the study of ``communication amid uncertainty'' has been the subject of a series of recent work starting with Goldreich, Juba and Sudan \cite{JS1,GJS:jacm} followed by \cite{JKKS11,JS2,JubaWilliams13,HaramatyS,CGMS15}. While early works were very abstract and general, later works (starting with Juba, Kalai, Khanna and Sudan \cite{JKKS11}) tried to explore the ramifications of uncertainty in Yao's standard communication complexity model \cite{Yao}. In particular, the more recent works relax the different pieces of context that were assumed to be \emph{perfectly shared} in Yao's model, such as shared randomness \cite{CGMS15}, and in a recent work of the authors with Komargodski and Kothari the \emph{function being computed} \cite{ghazi2015communication}.

Specifically, \cite{ghazi2015communication} study the following functional notion of uncertainty in communication. Their setup builds on -- and generalizes -- Yao's classical model of (distributional) communication complexity, where Alice has an input $x$
and Bob has an input $y$, with $(x,y)$ being sampled from a distribution $\mu$. Their
goal is to communicate minimally so as to compute some function $g(x,y)$ (with
high probability over the choice of $(x,y)$). The understated emphasis of the
model is that for many functions $g$, the communication required is much less
than the lengths of $x$ or $y$, the entropy of $x$ or $y$ or even the conditional entropy of $x$ given $y$.

The question studied by \cite{ghazi2015communication} is: How much of this gain in communication is preserved
when the communicating parties do not exactly agree on the function being computed? (We further discuss the importance of this question in Section~\ref{subsec:imp}.)
This variation of the problem is modelled as follows: Alice is given a Boolean function $f$ and an input string $x$, and Bob is given a Boolean function $g$ and an input string $y$ where $(x,y)$ is sampled from a known distribution $\mu$ as
before, and $(f,g)$ is chosen (adversarially) from a known class $\mathcal{F}$ of pairs of functions that are close in terms of the Hamming distance $\Delta_{\mu}$ (weighted according to $\mu$). 
Alice and Bob wish to compute $g(x,y)$. 
Alice's knowledge of the function $f$ (which is close but not necessarily equal to $g$) captures the uncertainty in the knowledge of the context.

We define the \emph{public-coin uncertain communication complexity} $\CCU^{\mu}_{\epsilon}(\mathcal{F})$ as the minimum length of a two-way public-coin protocol whose output is correct with probability at least $1-\epsilon$ over its internal randomness and that of $(x,y)$. We similarly define the \emph{private-coin uncertain communication complexity} $\PrivCCU^{\mu}_{\epsilon}(\mathcal{F})$ by restricting to private-coin protocols. Clearly, $\CCU^{\mu}_{\epsilon}(\mathcal{F}) \le \PrivCCU^{\mu}_{\epsilon}(\mathcal{F})$. The quantities $\owCCU^{\mu}_{\epsilon}(\mathcal{F})$ and $\owPrivCCU^{\mu}_{\epsilon}(\mathcal{F})$ are similarly defined by restricting to one-way protocols.\footnote{Note that the uncertain model is clearly a generalization of Yao's model which corresponds to the particular case where $\mathcal{F} = \{(f,f)\}$ for some fixed function $f$. On the other hand, the uncertain model can also be viewed as a particular case of Yao's model via an exponential blow-up in the input size. For more on this view (which turns out to be ineffective in our setup), we refer the reader to Note~\ref{rem:phil_obj} at the end of this section.\label{ft:phil_obj}}

The previous work (\cite{ghazi2015communication}) gave an upper bound on $\owCCU^{\mu}_{\epsilon}(\mathcal{F})$ whenever $\mathcal{F}$ consists of functions $g$ whose one-way distributional complexity is small. More precisely, denote by $\owCC^\mu_{\epsilon}(g)$ the one-way communication complexity of $g$ in the standard distributional model.\footnote{By the ``easy direction'' of Yao's min-max principle, we can without loss of generality consider deterministic (instead of public-coin) protocols when defining $\owCC^\mu_{\epsilon}(g)$. We point out that this is not true in the uncertain case.} Namely, $\owCC^\mu_{\epsilon}(g)$ is the minimum length of a one-way deterministic protocol computing $g$ with probability at least $1-\epsilon$ over the randomness of $(x,y)$. Then, \cite{ghazi2015communication} showed that if $\mathcal{F}$ consists of pairs $(f,g)$ of functions that are at distance $\delta$, and if $\owCC^\mu_{\epsilon}(f), \owCC^\mu_{\epsilon}(g) \le k$, then 
for every positive $\theta$, 
$\owCCU_{\epsilon + 2\delta + \theta}^{\mu}(\mathcal{F})
    \leq O_\theta(k\cdot (1 + I(x;y)))$, where
$I(x;y)$ denotes the mutual information between $x$ and $y$.\footnote{One interpretation of the dependence of the communication in the uncertain setup on the mutual information $I(x;y)$ is that the players are better able to use the correlation of their inputs in the standard case than in the uncertain case.}
Note that if $\mu$ is a product distribution and if we let the parameter $\theta$ be a small constant, then the blow-up in communication is only a constant factor. However, the protocol of \cite{ghazi2015communication} crucially uses public randomness, and one of the main motivations behind this work is to understand how large the blow-up would be in the case where Alice and Bob have access to weaker types of randomness (or no randomness at all).

We point out that understanding the type of randomness that is needed in order to cope with uncertainty is a core question in the setup of communication with contextual uncertainty: \emph{If Alice and Bob do not (perfectly) agree on the function being computed, why can we assume that they (perfectly) agree on the shared randomness?}

\subsection{Our Contributions}

We prove several results about the power of shared randomness in uncertain communication. 
\paragraph*{Private and Imperfectly Shared Randomness}
Our first  result (Theorem~\ref{th:det_sep}) shows that private-coin protocols are much weaker than public-coin protocols in the setup of communication with contextual uncertainty. Far from obtaining a constant factor blow-up in communication, 
private-coin protocols incur an increase that is a growing function of $n$ when dealing with uncertainty.

Let $\mathcal{U} \triangleq \mathcal{U}_{2 n}$ be the uniform distribution on $\{0,1\}^{2 \cdot n}$. For positive integers $t$ and $n$, we define $\log^{(t)}(n)$ by setting $\log^{(1)}(n) = \log{n}$, and $\log^{(i)}(n) = \max(\log{\log^{(i-1)}(n)},1)$ for all $i \in \{2,\dots,t\}$.
\begin{theorem}[Lower-bound on private-coin uncertain protocols]\label{th:det_sep}
For every sufficiently small $\delta >0$, there exist a positive integer $\ell \triangleq \ell(\delta)$ and a function class $\mathcal{F} \triangleq \mathcal{F}_{\delta}$ such that
\begin{enumerate}[label=(\roman*)]
\item\label{part:det_dist} For each $(f,g) \in \mathcal{F}$, we have that $\Delta_{\mathcal{U}}(f,g) \le \delta$.
\item\label{part:det_cert_ub} For each $(f,g) \in \mathcal{F}$, we have that $\owCC^{\mathcal{U}}_0(f), \owCC^{\mathcal{U}}_0(g) \le \ell$ .
\item\label{part:det_unc_lb} For every $\eta > 0$ and $\epsilon \in (4\delta, 0.5]$, we have that $\PrivCCU^{\mathcal{U}}_{\epsilon/2-2\delta-\eta}(\mathcal{F}) = \Omega(\eta^{2} \cdot \log^{(t)}(n))$ for some positive integer $t = \Theta((\epsilon/\delta)^2)$.
\end{enumerate}
\end{theorem}

In Theorem~\ref{th:det_sep}, the inputs $x$ and $y$ are binary strings of length $n$ and $\calf$ is a family of pairs of functions, which each function mapping $\{0,1\}^n \times \{0,1\}^n$ to $\{0,1\}$. Also, the parameter $\eta$ can possibly depend on $n$. We point out that Theorem~\ref{th:det_sep} also implies the first separation between \emph{deterministic} uncertain protocols and public-coin uncertain protocols\footnote{This uses the fact that private-coin communication complexity is no larger than deterministic communication complexity, both in the certain and uncertain setups.}.



\begin{note}
We point out that the relative power of private-coin and public-coin protocols in the uncertain model is both conceptually and technically different from the standard model. Specifically, the randomness is potentially used in the standard model in order to fool an adversary selecting the input pair $(x,y)$, whereas in the uncertain model, it is used to fool an adversary selecting the pair $(f,g)$ of functions that are promised to be close. This promise makes the task of proving lower bounds against private-coin protocols in the uncertain model (e.g., Theorem~\ref{th:det_sep}) significantly more challenging than in the standard model.\footnote{In particular, the diagonilization-based arguments that imply a separation between the public-coin and the private-coin communication complexities of the Equality function in the standard model completely fail when we impose such a promise.} Moreover, a well-known theorem due to Newman \cite{newman1991private} shows that in the standard model, any public-coin protocol can be simulated by a private-coin protocol while increasing the communication by an additive $O(\log{n})$ bits. By contrast, there is no known analogue of Newman's theorem in the uncertain case!
\end{note}

\begin{note}\label{rem:priv}
The construction that we use to prove Theorem~\ref{th:det_sep} cannot give a separation larger than $\Theta(\log{\log{n}})$. Thus, showing a separation of $\omega(\log{\log{n}})$ between private-coin and public-coin protocols in the uncertain case would require a new construction. For more details, see Note~\ref{rem:priv_elab}.
\end{note}

In light of Theorem~\ref{th:det_sep}, it is necessary for Alice and Bob to share some form of randomness in order to only incur a constant blow-up in communication for product distributions. Fortunately, it turns out that it is not necessary for Alice and Bob to \emph{perfectly} share a sequence of random coins. If Alice is given a uniform-random string $r$ of bits and Bob is given a string $r'$ obtained by independently flipping each coordinate of $r$ with probability $0.49$, then efficient communication is still possible!

More formally, for $\rho \in [0,1]$, define $\owIsrCCU^{\mu}_{\epsilon,\rho}(\mathcal{F})$ in the same way that we defined $\owCCU^{\mu}_{\epsilon}(\mathcal{F})$ except that instead of Alice and Bob having access to public randomness, Alice will have access to a sequence $r$ of independent uniformly-random bits, and Bob will have access to a sequence $r'$ of bits obtained by independently flipping each coordinate of $r$ with probability $(1-\rho)/2$. Note that this setup of imperfectly shared randomness interpolates between the public randomness and private randomness setups, i.e., $\owIsrCCU^{\mu}_{\epsilon,1}(\mathcal{F}) = \owCCU^{\mu}_{\epsilon}(\mathcal{F})$ and $\owIsrCCU^{\mu}_{\epsilon,0}(\mathcal{F}) = \owPrivCCU^{\mu}_{\epsilon}(\mathcal{F})$.

\begin{theorem}[Uncertain protocol using imperfectly shared randomness]\label{le:isr}
Let $\rho \in (0,1]$ and $\mu$ be a \emph{product} distribution. Let $\mathcal{F}$ consist of pairs $(f,g)$ of functions with $\Delta_{\mu}(f,g) \le \delta$, and $\owCC^\mu_{\epsilon}(f), \owCC^\mu_{\epsilon}(g) \le k$. Then, for every positive $\theta$, 
$\owIsrCCU_{\epsilon + 2\delta + \theta, \rho}^{\mu}(\mathcal{F}) \leq O_\theta(k/\rho^2)$.
\end{theorem}

The \emph{imperfectly shared randomness} model in Theorem~\ref{le:isr} was recently independently introduced (in the setup of communication complexity) by Bavarian, Gavinsky and Ito \cite{bavarian2014role} and by Canonne, Guruswami, Meka and Sudan \cite{CGMS15} (and it was further studied in \cite{ghazi2015communicationperm}). Moreover, our proof of Theorem~\ref{le:isr} is based on combining the uncertain protocol of \cite{ghazi2015communication} and the locality-sensitive-hashing based protocol of \cite{CGMS15}.

We point out that Theorem~\ref{le:isr} also holds for more general i.i.d. sources of correlated randomness than the one described above. More precisely, for i.i.d. (not necessarily binary) sources of (imperfectly) shared randomness with \emph{maximal correlation}\footnote{The \emph{maximal correlation} of a pair $(X,Y)$ of random variables (with support $\mathcal{X} \times \mathcal{Y}$) is defined as $\rho(X,Y) \triangleq \sup \Ex[F(X) G(Y)]$ where the supremum is over all functions $F:\mathcal{X} \to \mathbb{R}$ and $G:\mathcal{Y} \to \mathbb{R}$ with $\Ex[F(X)] = \Ex[G(Y)] = 0$ and $\Var[F(X)] = \Var[G(Y)] =  1$. It is not hard to show that the binary source of imperfectly shared randomness defined in the paragraph preceding Theorem~\ref{le:isr} has maximal correlation $\rho$.}  $\rho$, the work of Witsenhausen \cite{witsenhausen1975sequences} along with the protocols of \cite{CGMS15} and \cite{ghazi2015communication} imply an uncertain protocol with $O_{\theta}(k/\rho^2)$ bits of communication.

\paragraph*{Public Randomness}
We now turn to our next result where we consider the dependence of 
the upper bound of \cite{ghazi2015communication}
on the mutual information $I \triangleq I(X;Y)$ in the case of public-coin protocols. The previous work \cite{ghazi2015communication} proved a lower bound of $\Omega(\sqrt{I})$ on this dependence, but their lower-bound does not grow with $k$. We improve this lower bound to $\Omega(\sqrt{k} \cdot \sqrt{I})$.
\begin{theorem}[Improved lower-bound on public-coin uncertain protocols]\label{th:imp_rand_sep}
For every sufficiently small $\delta > 0$ and every positive integers $k,n$ such that $k = o(\exp(\sqrt{n}))$, there exist an input distribution $\mu$ on input pairs $(X,Y) \in \{0,1\}^{k \cdot n} \times \{0,1\}^{k \cdot n}$ with mutual information $I \approx k \cdot n$ and a function class $\mathcal{F} \triangleq \mathcal{F}_{\delta,k,n}$ such that\footnote{We note that $I \approx k \cdot n$ means that $I/(k \cdot n) \to 0$ as $n \to \infty$.}
\begin{enumerate}[label=(\roman*)]
\item\label{part:rand_dist} For each $(f,g) \in \mathcal{F}$, we have that $\Delta_{\mu}(f,g) \le \delta$.
\item\label{part:cert_ub} For each $(f,g) \in \mathcal{F}$, we have that $\owCC^\mu_0(f), \owCC^\mu_0(g) \le k$.
\item\label{part:uncert_lb} $\owCCU^{\mu}_{\epsilon}(\mathcal{F}) = \Omega(\sqrt{k} \cdot \sqrt{I})$ for some absolute constant $\epsilon > 0$ independent of $\delta$. 
\end{enumerate}
\end{theorem}
As will be explained in detail in Section~\ref{subsec:over_rand}, the proof of Theorem~\ref{th:imp_rand_sep} is based on an extension of the lower bound construction of \cite{ghazi2015communication}, which is then analyzed using different techniques.

\begin{note}\label{rem:pub}
The construction that we use to prove Theorem~\ref{th:imp_rand_sep} cannot give a lower bound larger than $\tilde{\Theta}(\sqrt{k} \cdot \sqrt{ I})$. Thus, improving on the lower bound in Theorem~\ref{th:imp_rand_sep} by more than logarithmic factors in $k$ and $I$ would require a new construction.
\end{note}

\paragraph*{New Communication Problems}
Our lower bounds in Theorems~\ref{th:det_sep} and \ref{th:imp_rand_sep} 
are derived by defining new problems in {\em standard} communication complexity
(i.e., without uncertainty) and proving lower bounds for these problems.
We describe these problems and our results on these next.

The construction that we use to prove Theorem~\ref{th:det_sep} requires us to understand the following ``subset-majority with side information'' setup. Alice is given a subset $S \subseteq [n]$ and a string $x \in \{ \pm 1\}^n$, and Bob is given a subset $T \subseteq [n]$ and a string $y \in \{\pm 1\}^n$. The subsets $S$ and $T$ are adversarially chosen but are promised to satisfy $S \subseteq T$, $|T| = \ell$ and $|T \setminus S| \le \delta \cdot \ell$ for some fixed parameters $\ell$ and $\delta$. The strings $x$ and $y$ are chosen independently and uniformly at random. Alice and Bob wish to compute the function $\SubsetMaj((S,x),(T,y)) \triangleq \Sign(\sum_{i \in T} x_i y_i)$. In words, $\SubsetMaj((S,x),(T,y))$ is equal to $0$ if $x$ and $y$ differ on a majority of the coordinates in subset $T$, and $1$ otherwise. Note that $S$ does not directly appear in the definition of the function $\SubsetMaj$ but it can serve as useful side-information for Alice.\footnote{Note that we could have alternatively defined $\SubsetMaj$ in terms of $S$, and let $T$ serve as the potentially useful side-information. Our lower bound would also apply to this setup.} What is the private-coin communication complexity of computing $\SubsetMaj$ on every $(S,T)$-pair satisfying the above promise and with high probability over the random choice of $(x,y)$ and over the private randomness? We prove the following (informally stated) lower bound.

\begin{lem}\label{le:det_lb_std_quest}
Any private-coin protocol computing $\SubsetMaj$ on every $(S,T)$-pair satisfying the promise and with high probability over the random choice of $(x,y)$ and over the private randomness should communicate at least $\log^{(t)}(n)$ bits for some positive integer $t$ that depends on $\delta$ and the error probability.
\end{lem}

The proof of Theorem~\ref{th:imp_rand_sep} is based on a construction that leads to the question described next regarding the communication complexity of a particular block-composed function. Namely, consider the following ``majority composed with subset-parity with side information'' setup. Alice is given a sequence of subsets $S \triangleq (S^{(i)} \subseteq [n])_{i \in [k]}$ and a sequence of strings $x \triangleq (x^{(i)} \in \{0,1\}^n)_{i \in [k]}$, and Bob is given a sequence of subsets $T \triangleq (T^{(i)} \subseteq [n])_{i \in [k]}$ and a sequence of strings $y \triangleq (y^{(i)} \in \{0,1\}^n)_{i \in [k]}$. We consider the following distribution $\mu$ on $((S,x),(T,y))$. Independently for each $i \in [k]$, we sample $((S^{(i)},x^{(i)}),(T^{(i)},y^{(i)}))$ as follows: $S^{(i)}$ is a uniform-random subset and $T^{(i)}$ is an $\epsilon$-noisy\footnote{This means that the indicator vector of $T^{(i)}$ is obtained by independently flipping each coordinate of the indicator vector of $S^{(i)}$ with probability $\epsilon$.} copy of $S^{(i)}$, and independently $x^{(i)}$ is a uniform-random string and $y^{(i)}$ is an $\epsilon$-noisy copy of $x$. Here, $\epsilon$ is a positive parameter that can depend on $n$ and $k$. Alice and Bob wish to compute the function $\Maj \circ \SubsetParity((S,x),(T,y)) \triangleq \Sign\big( \sum_{i=1}^k (-1)^{\langle T^{(i)}, x^{(i)} \oplus y^{(i)} \rangle}\big)$ where $T^{(i)}$ denotes both the subset and its $0/1$ indicator vector, the inner product is over $\mathbb{F}_2$, and $x^{(i)} \oplus y^{(i)}$ is the coordinate-wise XOR of $x^{(i)}$ and $y^{(i)}$. What is the communication complexity of computing $\Maj \circ \SubsetParity$ with high probability over the distribution $\mu$? We prove the following lower bound.

\begin{lem}\label{le:rand_lb_std_quest}
Any $1$-way protocol computing $\Maj \circ \SubsetParity$ with high probability over the distribution $\mu$ should communicate $\Omega(k \cdot \epsilon \cdot n)$ bits.
\end{lem}


In Section~\ref{subsec:over_det}, we outline the proof of Theorem~\ref{th:det_sep} and explain how it leads to the setup of Lemma~\ref{le:det_lb_std_quest} and how we prove Lemma~\ref{le:det_lb_std_quest}. In Section~\ref{subsec:over_rand}, we outline the proof of Theorem~\ref{th:imp_rand_sep} and explain how it leads to the setup of Lemma~\ref{le:rand_lb_std_quest} and how we prove Lemma~\ref{le:rand_lb_std_quest}.

Before doing so, we discuss some conceptual implications of our results.

\subsection{Implications}\label{subsec:imp}

Functional uncertainty models much of the day-to-day interactions among humans, where a person is somewhat aware of the objectives of the other person she is interacting with, but do not know them precisely. Neither person typically knows exactly what aspects of their own knowledge may be relevant to the interaction,  yet they do manage to have a short conversation. This is certainly a striking phenomenon, mostly unexplained in mathematical terms. This line of works aims to explore such phenomena. It is important to understand what mechanisms may come into play here, and what features play a role. Is the ability to make random choices important? Is shared information crucial? Is there a particular measure of distance between functions that makes efficient communication feasible? In order to understand such questions, one first needs to have a ground-level understanding of communication with functional uncertainty. This work tackles several basic questions that remain unexplored.

An ideal model for communication would only assume a constant amount of perfectly shared context between the sender and receiver, such as the knowledge of an encoding/decoding algorithm, one universal Turing machine, etc. Solutions to most interesting communication problems seem to assume a shared information which grows with the length of the inputs. Recent work showed that in many of these scenarios some assumptions about the shared context can be relaxed to an imperfect sharing, but these results are often brittle and break when two or more contextual elements are simultaneously assumed to be imperfectly shared. Our work raises the question of whether {\em imperfectly shared randomness} would be sufficient to overcome {\em functional uncertainty}. We show that this is indeed the case for product distributions, but the loss for non-product distributions might be much larger (for this and other open questions, we refer the reader to our conclusion Section~\ref{sec:conc}). Such results highlight the
delicate nature of the role of shared context in communication. They
beg for a more systematic study of communication
which at the very least should be
able to mimic the aims, objectives and phenomena
encountered in human communication.

\begin{note}\label{rem:phil_obj}
As mentioned in Footnote~\ref{ft:phil_obj}, the uncertain model is clearly a generalization of Yao's model. Strictly speaking, the uncertain model can also be viewed as a particular case of Yao's model by regarding the function(s) that is being computed as part of the inputs of Alice and Bob, which results in an exponential blow-up in the input-size. This latter view turns out to be fruitless for our purposes. Indeed, from this perspective, all the different well-studied communication functions (such as Equality, Set Disjointness, Pointer Jumping, etc.) are regarded as special cases of one ``universal function''! More importantly, this view completely blurs the distinction between the \emph{goal} of the communication (i.e., the function to compute) and the inputs of the parties. On a technical level, it does not simplify the task of proving the lower bounds in Theorems~\ref{th:det_sep} and \ref{th:imp_rand_sep} in any way since it does not capture the promise that the two functions (given to Alice and Bob) are close in Hamming distance. Thus, in the rest of this paper, we stick to the former view and use the expressions ``uncertain model'' and ``standard model" to refer to the setups with and without uncertainty, respectively.
\end{note}

\section{Overview of Proofs}
\subsection{Overview of Proof of Theorem~\ref{th:det_sep}}\label{subsec:over_det}
\paragraph*{Reduction to Lemma~\ref{le:det_lb_std_quest}.} 
In order to prove Theorem~\ref{th:det_sep}, we need to devise a function class for which circumventing the uncertainty is much easier using public randomness than using private randomness. One general setup in which Bob can leverage public randomness to resolve some uncertainty regarding Alice's knowledge is the following ``small-set intersection'' problem. Assume that Alice is given a subset $S \subseteq [n]$, and Bob is given a subset $T \subseteq [n]$ such that $T$ contains $S$ and $|T| = \ell$, where we think of $\ell$ as being a large constant. Here, Bob knows that Alice has a subset of his own $T$ but he is uncertain which subset Alice has. Using public randomness, a standard $1$-way hashing protocol communicating $\tilde{O}(\ell)$ bits allows Bob to determine $S$ with high probability. On the other hand, using only private randomness, the communication complexity of this task is $\Theta(\log\log{n})$ bits.

With the above general setup in mind, we consider functions $f_S$ indexed by small subsets $S$ of coordinates on which they depend. Since we want the functions $f_S$ and $f_T$ to be close in Hamming distance, we enforce $|T \setminus S|$ to be small for every pair $(f_S,f_T)$ of functions in our class, and we let each function $f_S$ be ``noise-stable''. Since we want our function $f_S$ to genuinely depend on all coordinates in $S$, the majority function $f_S(x,y) = \Sign(\sum_{i \in S} x_i y_i)$ for $x,y \in \{\pm 1\}^n$ arises as a natural choice. We also let $x$ and $y$ be independent uniform-random strings. In this case, it can be seen that if $|T \setminus S|$ is a small constant fraction of $|T|$, then the quadratic polynomials $\sum_{i \in S} x_i y_i$ and $\sum_{i \in T} x_i y_i$ behave like standard Gaussians with correlation close to $1$, and the quadratic threshold functions $f_S(x,y)$ and $f_T(x,y)$ are thus close in Hamming distance.

Note that in the certain case, i.e., when both Alice and Bob agree on $S$, they can easily compute $f_S(x,y)$ by having Alice send to Bob the $\ell$ bits $(x_i)_{i \in S}$. Moreover, if Alice and Bob are given access to public randomness in the uncertain case, Bob can figure out $S$ via the hashing protocol mentioned above using $\tilde{O}(\ell)$ bits of communication, which would reduce the problem to the certain case\footnote{Alternatively, Alice and Bob can run the protocol of \cite{ghazi2015communication} which would communicate $O(\ell)$ bits.}. The bulk of the proof will be to lower-bound the private-coin uncertain communication. Note that by the choice of our function class and distribution, this is equivalent to proving 
Lemma~\ref{le:det_lb_std_quest}. 

\paragraph*{Proof of Lemma~\ref{le:det_lb_std_quest}.} 
To prove Lemma~\ref{le:det_lb_std_quest},
the high-level intuition is that a protocol solving the uncertain problem should be essentially revealing to Bob the subset $S$ that Alice holds. Formalizing this intuition turns out to be challenging, especially that a private-coin protocol solving the uncertain problem is only required to output a \emph{single bit} which is supposed to equal the Boolean function $f_T(x,y)$ with high probability over $(x,y)$ and over the private randomness. In fact, this high-level intuition can be shown not to hold in certain regimes\footnote{For example, for constant error probabilities, the $1$-way randomized communication complexity of small-set intersection is known to be $\Theta(\ell\cdot\log(\ell))$ bits (see, e.g., \cite{brody2014beyond}) whereas the public-coin protocol of \cite{ghazi2015communication} can compute $f_T$ with $O(\ell)$ bits of communication.}. Moreover, the standard proofs that lower bound the communication of small-set intersection do not extend to lower-bound the communication complexity of $f_T$.

To lower-bound the private-coin communication of solving the uncertain task by a growing function of $n$, we consider the following \emph{shift communication game}. Bob is given a sorted tuple $\sigma = (\sigma_1,\dots,\sigma_t)$ of integers with $1 \le \sigma_1 < \dots < \sigma_t \le n$, and Alice is either given the prefix $(\sigma_1,\dots,\sigma_{t-1})$ of length $t-1$ of $\sigma$ or the suffix $(\sigma_2,\dots,\sigma_t)$ of length $t-1$ of $\sigma$. Bob needs to determine the input of Alice. We show that a celebrated lower bound of Linial \cite{linial1992locality} on the \emph{chromatic number} of certain related graphs implies a lower bound of $\log^{(t+1)}(n)$ on the private-coin communication of the shift communication game. We then show that any private-coin protocol solving the uncertain task can be turned into a private-coin protocol solving the shift-communication game with a constant (i.e., independent of $n$) blow-up in the communication (see Protocol~\ref{alg:red_prot}).

\subsection{Overview of Proof of Theorem~\ref{th:imp_rand_sep}}\label{subsec:over_rand}
\paragraph*{Reduction to Lemma~\ref{le:rand_lb_std_quest}.}
The proof of Theorem~\ref{th:imp_rand_sep} builds on the lower-bound construction of \cite{ghazi2015communication} which we recall next. Let $\mu$ be the distribution over pairs $(x,y) \in \{0,1\}^{2n}$ where $x$ is uniform-random and $y$ is an $\epsilon$-noisy copy of $x$ with $\epsilon = \sqrt{\delta/n}$. Then, the mutual information between $x$ and $y$ satisfies $I \approx n$. For each $S \subseteq [n]$, consider the function $f_S(x,y) \triangleq \langle S, x \oplus y \rangle$ where the inner product is over $\mathbb{F}_2$, $x \oplus y$ denotes the coordinate-wise XOR of $x$ and $y$, and $S$ is used to denote both the subset and its $0/1$ indicator vector. Moreover, consider the class $\mathcal{F}$ of all pairs of functions $(f_S,f_T)$ where $|S \triangle T| \le \sqrt{\delta n}$. It can be seen that for such $S$ and $T$, the distance between $f_S$ and $f_T$ under $\mu$ is at most $\delta$. If Alice and Bob both know $S$, then Alice can send the single bit $\langle S, x \rangle$ to Bob who can then output the correct answer $\langle S, x \oplus y \rangle = \langle S, x \rangle \oplus \langle S, y \rangle$. This means that the certain communication is $1$ bit. Using the well-known discrepancy method, \cite{ghazi2015communication} showed a lower bound of $\Omega(\sqrt{n})$ bits on the communication of the associated uncertain problem. Since in this case $I \approx n$, this in fact lower-bounds the uncertain communication by $\Omega(\sqrt{I})$ bits. For this construction, this lower bound turns out to be tight up to a logarithmic factor.

To improve the lower-bound from $\sqrt{I}$ to $\sqrt{k} \cdot \sqrt{I}$, we consider the following ``block-composed'' framework. Let $\{f_{S^{(i)}}(x^{(i)}, y^{(i)}): i \in [k]\}$ be $k$ independent copies of the above base problem of \cite{ghazi2015communication} and consider computing the composed function $g\big(f_{S^{(1)}}(x^{(1)}, y^{(1)}), \dots, f_{S^{(k)}}(x^{(k)}, y^{(k)})\big)$ for some outer function $g:\{0,1\}^k \to \{0,1\}$. For any choice of $g$, the certain communication of the composed function would be at most $k$ bits. When choosing the outer function $g$ to use in our lower bound, we thus have two objectives to satisfy. First, $g$ has to be sufficiently \emph{hard} in the sense that its average-case decision tree complexity with respect to the uniform distribution on $\{0,1\}^k$ should be $\Omega(k)$; otherwise, it will not be the case that the uncertain communication of computing $g$ on $k$ copies of the base problem is at least $k$ times the uncertain communication of the base problem. Second, $g$ has to be \emph{noise stable} in order to be able to upper bound the distance between $g\big(f_{S^{(1)}}(\cdot), \dots, f_{S^{(k)}}(\cdot)\big)$ and $g\big(f_{T^{(1)}}(\cdot), \dots, f_{T^{(k)}}(\cdot)\big)$.

Note that setting $g$ to be a dictator function would satisfy the noise-stability property, but it clearly would not satisfy the hardness property, as the composed function would be equal to the base function and would thus have uncertain communication $\tilde{O}(\sqrt{n})$ bits. Another potential choice of $g$ is to set it to the parity function on $k$ bits. This function would satisfy the hardness property, but it would strongly violate the noise stability property that is crucial to us. This leads us to setting $g$ to the majority function on $k$ bits, which is well-known to be noise stable, and has average-case decision-tree complexity $\Omega(k)$ with respect to the uniform distribution on $\{0,1\}^k$. In fact, the noise stability of the majority function readily implies an upper bound of $O(\sqrt{\delta})$ on the distance between any pair of composed functions that are specified by tuples of subsets $(S^{(1)}, \dots, S^{(k)})$ and $(T^{(1)}, \dots, T^{(k)})$ with $|S^{(i)} \triangle T^{(i)}| \le \sqrt{\delta n}$ for each $i \in [k]$. The crux of the proof will be to lower-bound the uncertain communication of the majority-composed function by $\Omega(k \sqrt{n})$, which amounts to proving Lemma~\ref{le:rand_lb_std_quest}. Since in this block-composed framework the mutual information satisfies $I \approx kn$, this would imply the lower bound of $\Omega(\sqrt{k} \sqrt{I})$ on the uncertain communication in Part~\ref{part:uncert_lb} of Theorem~\ref{th:imp_rand_sep}.

\paragraph*{Proof of Lemma~\ref{le:rand_lb_std_quest}.} 
We first point out that the average-case \emph{quantum} decision tree complexity of $\Maj_k$ with respect to the uniform distribution is $\tilde{O}(\sqrt{k})$ \cite{ambainis2001average}. This implies that any communication complexity lower-bound method that extends to the quantum model cannot prove a lower bound larger than $\tilde{O}(\sqrt{k} \cdot \sqrt{n})$ on our uncertain communication\footnote{Thus, since $I \approx k\cdot n$ in our block-composed framework, such methods cannot be used to improve the lower-bound of $\Omega(\sqrt{I})$ of \cite{ghazi2015communication} by more than logarithmic factors.}. In particular, we cannot solely rely on the discrepancy bound (as done in \cite{ghazi2015communication}), since this bound is known to lower-bound quantum communication. Similarly, the techniques of \cite{sherstov2008pattern,shi2007quantum,lee2010composition} rely on the generalized discrepancy bound (originally due to \cite{klauck2001lower}) which also lower-bounds quantum communication. Moreover, the recent results of \cite{molinaro2015amplification} only apply to product distributions (i.e., where Alice's input is independent of Bob's input) in contrast to our case where the inputs of Alice and Bob are very highly-correlated. Finally, the recent works of \cite{goos2015rectangles,goos2015landscape} do not imply lower bounds on the average-case complexity with respect to the distribution that arises in our setup.

To circumvent the above obstacles, we use a new approach that is tailored to our setup and that is outlined next. Let $\Pi$ be a $1$-way protocol solving the uncertain task with high probability. We consider the information that $\Pi$ reveals about the \emph{inputs to the outer function}, i.e., about the length-$k$ binary string $\big(f_{S^{(1)}}(x^{(1)}, y^{(1)}), \dots, f_{S^{(k)}}(x^{(k)}, y^{(k)})\big)$. We call this quantity the \emph{intermediate information cost} of $\Pi$, and we argue that it is at least $\Omega(k)$ bits. To do so, we recall the Hamming distance function $\HD_k$ defined by $\HD_k(u,v) = 1$ if the Hamming distance between $u$ and $v$ is at least $k/2$ and $\HD_k(u,v) = 0$ otherwise. We upper bound the information complexity of computing $\HD_k$ over the uniform distribution on $\{0,1\}^{2k}$ by the intermediate information cost of $\Pi$. We do so by giving an information-cost preserving procedure (Protocol~\ref{alg:sim_prot}) where Alice and Bob are given independent uniformly distributed $u$ and $v$ (respectively) and use their private and public coins in order to simulate the input distribution $(X,Y)$ of our uncertain problem. The known $1$-way lower bound of \cite{woodruff2007efficient} on $\HD_k$ under the uniform distribution then implies that $\Pi$ reveals $\Omega(k)$ bits  of information to Bob about the tuple $\big(f_{S^{(1)}}(x^{(1)}, y^{(1)}), \dots, f_{S^{(k)}}(x^{(k)}, y^{(k)})\big)$. This allows Bob to guess this tuple with probability $0.51^k$. We then apply the strong direct product theorem for discrepancy of \cite{lee2008direct} which, along with the discrepancy-based lower bound on the communication of the base uncertain problem of \cite{ghazi2015communication}, implies that $\Pi$ should be communicating at least $\Omega(k \sqrt{n})$ bits.

\paragraph*{Organization of the rest of the paper}
In Section~\ref{sec:prelim}, we give some preliminaries that will be useful to us. In Section~\ref{sec:det_sep}, we prove Theorem~\ref{th:det_sep}. In Section~\ref{sec:rand_imp_sep}, we prove Theorem~\ref{th:imp_rand_sep}. The proof of Theorem~\ref{le:isr} is given in Section~\ref{sec:isr_app_pf}. In Section~\ref{sec:conc}, we conclude with some interesting open questions. A useful lemma that is used in Section~\ref{sec:det_sep} appears in Appendix~\ref{sec:berry_esseen}.

\section{Preliminaries}\label{sec:prelim}

For a real number $x$, we define $\Sign(x)$ to be $1$ if $x \geq 0$ and $0$ if $x<0$. For a set $S$, we write $X \in_R S$ to indicate that $X$ is a random variable that is uniformly distributed on $S$. For a positive integer $n$, we let $[n] \triangleq \{1,\dots,n\}$. For a real number $x$, we denote $\exp(x)$ a quantity of the form $2^{\Theta(x)}$. For any two subsets $S, T \subseteq [n]$, we let $S \setminus T$ be the set of all elements of $S$ that are not in $T$. We let $S \triangle T$ be the symmetric difference of $S$ and $T$, i.e., the union of $S \setminus T$ and $T \setminus S$. For functions $f,g:\mathcal{X} \times \mathcal{Y} \to \{0,1\}$ and any distribution $\mu$ on $\mathcal{X} \times \mathcal{Y}$, we define the distance $\Delta_{\mu}(f,g) \triangleq \Pr_{(x,y) \sim \mu}[f(x,y) \neq g(x,y)]$ as the Hamming distance between the values of $f$ and $g$, weighted with respect to $\mu$. If $\mu$ is the uniform distribution on $\mathcal{X} \times \mathcal{Y}$, we drop the subscript $\mu$ and denote $\Delta_{\mu}$ by $\Delta$. We next recall the standard communication complexity model of Yao \cite{Yao}.
 In Definitions~\ref{def:det_cc}, \ref{def:priv_cc} and \ref{def:dist_cc} below, we let $f\colon \mathcal{X} \times \mathcal{Y} \rightarrow \zo$ be a Boolean function.

\begin{definition}[Deterministic Communication Complexity]\label{def:det_cc}
	The \emph{two-way (resp. one-way) deterministic communication complexity} of $f$, denoted by $\CC(f)$ (resp. $\owCC(f)$), is defined as the minimum over all two-way (resp. one-way) deterministic protocols $\Pi$ that compute $f$ correctly on every input pair, of the communication cost of $\Pi$.
\end{definition}

\begin{definition}[Private-Coin Communication Complexity]\label{def:priv_cc}
	The \emph{two-way (resp. one-way) private-coin communication complexity} of $f$ with error $\epsilon$, denoted by $\PrivCC_{\epsilon}(f)$ (resp. $\owPrivCC_{\epsilon}(f)$), is defined as the minimum over all two-way (resp. one-way) private-coin protocols $\Pi$ that compute $f$ with probability at least $1-\epsilon$ on every input pair, of the communication cost of $\Pi$.
\end{definition}

The quantities in Definitions~\ref{def:det_cc} and \ref{def:priv_cc} can be similarly defined for partial functions $f$.

\begin{definition}[Distributional Communication Complexity]\label{def:dist_cc}
	Let $\mu$ be a distribution over $\mathcal{X} \times \mathcal{Y}$. The \emph{two-way (resp. one-way) distributional
		communication complexity} of $f$ over $\mu$ with error $\epsilon$, denoted by
	$\CC^{\mu}_\epsilon(f)$ (resp. $\owCC^{\mu}_\epsilon(f)$), is the minimum over all two-way (resp. one-way) protocols $\Pi$ that compute $f$ with probability $1-\epsilon$ over $\mu$, of the communication cost of $\Pi$.
\end{definition}

We next recall the model of communication with contextual uncertainty. For more details on this model, we refer the reader to \cite{ghazi2015communication}. In this setup, Alice knows a function $f$ and is given an input $x$, and Bob knows a function $g$ and is given an input $y$. Let $\calf \subseteq \{f\colon \mathcal{X} \times \mathcal{Y} \to \zo\}^2$ be a family of pairs of Boolean functions with domain $\mathcal{X}\times \mathcal{Y}$, and $\mu$ be a distribution on $\mathcal{X}\times \mathcal{Y}$. We say that a public-coin (resp. private-coin) protocol $\Pi$ $\epsilon$-computes $\calf$ over $\mu$ if for every $(f,g) \in \calf$, we have that $\Pi$ outputs the value $g(x,y)$ with probability at least $1-\epsilon$ over the randomness of $(x,y) \sim \mu$ and over the public (resp. private) randomness of $\Pi$.

\begin{definition}[Contextually Uncertain Communication Complexity]\label{def:ucc_gen_sepl}
  Let $\mu$ be a distribution on $\mathcal{X} \times \mathcal{Y}$ and
$\calf \subseteq \{f\colon \mathcal{X} \times \mathcal{Y} \to \zo\}^2$. The \emph{two-way (resp. one-way) public-coin communication complexity of $\calf$ under contextual uncertainty}, denoted $\CCU^\mu_{\epsilon}(\calf)$ (resp. $\owCCU^\mu_{\epsilon}(\calf)$), is the minimum over all two-way (resp. one-way) public-coin protocols $\Pi$
that $\epsilon$-compute $\calf$ over $\mu$, of the maximum communication
complexity of $\Pi$ over $(f,g) \in \calf$, $(x,y)$ from the support of
$\mu$ and settings of the public coins.

Similarly, the \emph{two-way (resp. one-way) private-coin communication complexity of $\calf$ under contextual uncertainty} $\PrivCCU^{\mu}_{\epsilon}(\mathcal{F})$ (resp. $\owPrivCCU^{\mu}_{\epsilon}(\mathcal{F})$) is defined by restricting to two-way (resp. one-way) private-coin protocols.
\end{definition}

\section{Construction for Private-Coin Uncertain Protocols}\label{sec:det_sep}

We now describe the construction that is used to prove Theorem~\ref{th:det_sep}. Each function in our universe is specified by a subset $S \subseteq [n]$ and is of the form $f_S:\{\pm 1\}^n \times \{\pm 1\}^n \to \{ 0,1\}$ with $f_S(X,Y) \triangleq \Sign( \sum_{i \in S} X_i Y_i)$ for all $X,Y \in \{\pm 1\}^n$. The function class is then defined by
\begin{equation*}
\mathcal{F}_{\delta} \triangleq \{ (f_S,f_T): S \subseteq T, ~ |T| = \ell \text{ and } |T \setminus S| \le \delta' \cdot \ell \},
\end{equation*}
where $\delta' = \alpha \cdot \delta^2$ for some sufficiently small positive absolute constant $\alpha$, and $\ell = \ell(\delta)$ is a sufficiently large function of $\delta$. The input pair $(X,Y)$ is drawn from the uniform distribution on $\{\pm 1\}^{2n}$. We start with the proof of Part~\ref{part:det_dist} of Theorem~\ref{th:det_sep}. It essentially follows from the fact that the the polynomials $\sum_{i \in S} X_i Y_i$ and $\sum_{i \in T} X_i Y_i$ behave like zero-mean Gaussians with unit-variance and correlation $\sqrt{1-\delta'}$.

We will need the following well-known fact which follows from Sheppard's forumla \cite{sheppard1899application}.
\begin{fact}\label{fact:sheppard}
	If $(X,Y)$ is a pair of zero-mean Gaussians with correlation $\Ex[XY] = \rho$, then
	\begin{equation*}
	\Pr[\Sign(X) \neq \Sign(Y)] = \frac{\arccos(\rho)}{\pi}.
	\end{equation*}
\end{fact}

We now prove Part~\ref{part:det_dist} of Theorem~\ref{th:det_sep}.
\begin{proof}[Proof of Part~\ref{part:det_dist} of Theorem~\ref{th:det_sep}]
	Let $S \subseteq T \subseteq [n]$ be such that $|T| = \ell$ and $|T \setminus S| \le \delta' \cdot \ell$. For fixed $\ell$, the distance $\Delta(f_S,f_T)$ decreases when $|T \setminus S|$ decreases. So it suffices to upper bound $\Delta(f_S,f_T)$  when $|T \setminus S| = \delta' \cdot \ell$. Assume that the coordinates in $T$ are $1 \le t_1 < t_2 < \dots < t_{\ell}$. Then, we define the random vectors $X', Y' \in \{0,1\}^{\ell}$ as $X'_i = X_{t_i} Y_{t_i}$ for all $i \in [\ell]$, and $Y'_i = X_{t_i} Y_{t_i}$ if $t_i \in S$ and $Y'_i = 0$ if $t_i \notin S$. We will apply the two-dimensional Berry-Esseen Theorem~\ref{th:2d_berry_ess_bin_not_id} to $(X',Y')$. To do so, first note that the random pairs $(X'_1,Y'_1), (X'_2,Y'_2), \dots, (X'_{\ell},Y'_{\ell})$ are independent. Moreover, for every $i \in [\ell]$ such that $t_i \in S$, the covariance matrix of $(X_i,Y_i)$ is given by $\Sigma_i = \begin{bmatrix}
	1       & 1 \\
	1       & 1 
	\end{bmatrix}$. On the other hand, for $i \in [\ell]$ such that $t_i \notin S$, the convariance matrix of $(X_i,Y_i)$ is given by $\Sigma_i = \begin{bmatrix}
	1       & 0 \\
	0       & 0 
	\end{bmatrix}$. Hence, the average (across the $\ell$ coordinates) covariance matrix is given by $\Sigma = \ell^{-1}\cdot \displaystyle\sum\limits_{i \in [\ell]} \Sigma_i = \begin{bmatrix}
	1      & 1-\delta' \\
	1-\delta'       & 1-\delta' 
	\end{bmatrix}$. The smallest and largest eigenvalues of $\Sigma$ are respectively given by
	\begin{subequations}
		\begin{empheq}{align}
		\lambda &\triangleq \frac{1-\delta'}{2} - \frac{\sqrt{5\cdot(1-\delta')^2 - 2 \cdot (1-\delta') + 1}}{2} + \frac{1}{2},\label{eq:min_eig_Sigma}\\
		\Lambda &\triangleq \frac{1-\delta'}{2} + \frac{\sqrt{5 \cdot (1-\delta')^2 - 2 \cdot (1-\delta') + 1}}{2} + \frac{1}{2}.\nonumber\label{eq:max_eig_Sigma}
		\end{empheq}
	\end{subequations}
	In Equation~(\ref{eq:min_eig_Sigma}), it can be checked that for $\delta' \in (0,1)$, $\lambda > 0$. By the two-dimensional Berry-Esseen Theorem~\ref{th:2d_berry_ess_bin_not_id}, we get that
	\begin{align}
	\Delta(f_S,f_T) &= \Pr[\Sign( \displaystyle\sum\limits_{i \in S} X_i Y_i) \neq \Sign( \displaystyle\sum\limits_{i \in T} X_i Y_i)]\nonumber\\ 
	&= \Pr[\Sign( \displaystyle\sum\limits_{i \in [\ell]} X'_i) \neq \Sign( \displaystyle\sum\limits_{i \in [\ell]} Y'_i)]\nonumber\\ 
	&= \Pr[\Sign(X'') \neq \Sign(Y'')] \pm O\bigg(\frac{1}{\lambda^{3/2} \cdot \sqrt{\ell}}\bigg),\label{eq:to_be_scaled}
	\end{align}
	where $(X'',Y'')$ is a pair of zero-mean Gaussians with covariance matrix $\Sigma$. We can scale $Y''$ to make it have unit-variance; this does not change its mean or the probability in Equation~(\ref{eq:to_be_scaled}). The covariance matrix becomes:
	$\Sigma' = \begin{bmatrix}
	1      & \sqrt{1-\delta'} \\
	\sqrt{1-\delta'}       & 1
	\end{bmatrix}$. By Sheppard's formula (i.e., Fact~\ref{fact:sheppard}), we deduce that
	\begin{align*}
	\Delta(f_S,f_T) &=  \frac{\arccos(\sqrt{1-\delta'})}{\pi} \pm O\bigg(\frac{1}{\lambda^{3/2} \cdot \sqrt{\ell}}\bigg)\\ 
	&= O(\sqrt{\delta'})\\ 
	&\le \delta,
	\end{align*}
	where we have used the facts that $\arccos(1-x) = O(\sqrt{x})$ for small positive values of $x$, that $\delta' = \alpha \cdot \delta^2$ for a sufficiently small positive absolute constant $\alpha$, and that $\ell$ is be a sufficiently large function of $\delta$.
\end{proof}

We now give the proof of Part~\ref{part:det_cert_ub} of Theorem~\ref{th:det_sep}, which is quite immediate.
\begin{proof}[Proof of Part~\ref{part:det_cert_ub} of Theorem~\ref{th:det_sep}]
	When both Alice and Bob know the subset $T \subseteq [n]$ which satisfies $|T| \le \ell$, Alice can send the sequence $(X_j: j \in T)$ of at most $\ell$ bits to Bob who can then output $f_T(X,Y)$.
\end{proof}

To prove Part~\ref{part:det_unc_lb} of Theorem~\ref{th:det_sep}, the next definition -- which is based on the graphs studied by Linial \cite{linial1992locality}-- will be crucial to us.

\begin{definition}[Shift Communication Game $\mathcal{G}_{m,t}$]
Let $m$ and $t$ be positive integers with $t \le m$. In the communication problem $\mathcal{G}_{m,t}$, Bob is given a sorted tuple $\sigma = (\sigma_1,\dots,\sigma_t)$ of distinct integers with $1 \le \sigma_1 < \dots < \sigma_t \le m$. In the YES case, Alice is given the prefix $(\sigma_1,\dots,\sigma_{t-1})$ of length $t-1$ of $\sigma$. In the NO case, Alice is given the suffix $(\sigma_2,\dots,\sigma_t)$ of length $t-1$ of $\sigma$. Alice and Bob need to determine which of the YES and NO cases occurs.
\end{definition}

Lemma~\ref{le:1_way_priv_G_lb} lower-bounds the private-coin communication complexity of $\mathcal{G}_{m,t}$. Its proof uses Linial's lower bound on the chromatic number of related graphs.
\begin{lem}\label{le:1_way_priv_G_lb}
There is an absolute constant $c$ such that for every sufficiently small $\epsilon > 0$, we have that $\PrivCC_{\epsilon}(\mathcal{G}_{m,t}) \geq c \cdot \log^{(t+2)}(m)$.
\end{lem}

We prove Lemma~\ref{le:1_way_priv_G_lb} in Section~\ref{subsec:ow_priv_G_lb}.  The proof of Part~\ref{part:det_unc_lb} of Theorem~\ref{th:det_sep} -- which is the main part in the proof of Theorem~\ref{th:det_sep} -- is given in Section~\ref{subsec:pf_part_iv_det_sep}.

\subsection{Proof of Lemma~\ref{le:1_way_priv_G_lb}}\label{subsec:ow_priv_G_lb}

The following family of graphs was first studied by Linial (in the setup of distributed graph algorithms) \cite{linial1992locality}.
\begin{definition}[Shift Graph $G_{m,t}$]\label{def:shift_graph}
	Let $m$ and $t$ be positive integers with $t \le m$. In the graph $G_{m,t} = (V_{m,t},E_{m,t})$, the vertices are all sorted tuples $\sigma = (\sigma_1,\dots,\sigma_t)$ of distinct integers with $1 \le \sigma_1 < \dots < \sigma_t \le m$. Two such tuples $\sigma$ and $\pi$ are connected by an edge in $E_{m,t}$ if and only if either $(\sigma_1,\dots,\sigma_{t-1}) = (\pi_2,\dots,\pi_t)$ or $(\sigma_2,\dots,\sigma_t) = (\pi_1,\dots,\pi_{t-1})$.
\end{definition}

Recall that the chromatic number $\chi(G)$ of an undirected graph $G$ is the minimum number of colors needed to color its vertices such that no two adjacent vertices share the same color. The following theorem is due to Linial.

\begin{theorem}[\cite{linial1992locality}, Proof of Theorem $2.1$]\label{th:linial_lb}
	Let $m$ and $t$ be positive integers such that $t \le m$ and $t$ is odd. Then, $\chi(G_{m,t}) \geq \log^{(t-1)}(m)$.
\end{theorem}

The next lemma uses Theorem~\ref{th:linial_lb} to lower-bound the deterministic two-way communication complexity of the shift communication game $\mathcal{G}_{m,t}$.
\begin{lem}\label{le:ow_det_lb_Gmt}
	Let $m$ and $t$ be positive integers such that $t \le m$ and $t$ is odd. Then, it is the case that $\CC(\mathcal{G}_{m,t}) \geq \log^{(t+1)}(m)$.
\end{lem}

\begin{proof}[Proof of Lemma~\ref{le:ow_det_lb_Gmt}]
	Assume for the sake of contradiction that there exists a deterministic two-way protocol that computes $\mathcal{G}_{m,t}$ and that has communication cost smaller than $\log^{(t+1)}(m)$. Then, using the fact that the one-way communication complexity of any function is at most exponential in its two-way communication complexity, we get that there is a one-way protocol $\Pi$ that computes $\mathcal{G}_{m,t}$ and that has communication cost smaller than $\log^{(t)}(m)$. Let $M$ be the single message sent from Alice to Bob under $\Pi$. Then, the length of $M$ satisfies $|M| < \log^{(t)}(m)$. Note that Alice's input is an element of the vertex-set $V_{m,t}$ of the shift graph $G_{m,t}$ (Definition~\ref{def:shift_graph}). Since $M$ is a deterministic function of Alice's input, it induces a coloring of $V_{m,t}$ into less than $2^{\log^{(t)}(m)} = \log^{(t-1)}(m)$ colors. The fact that no two adjacent vertices in $G_{m,t}$ share the same color follows from the correctness of $\Pi$ in computing $\mathcal{G}_{m,t}$. This contradicts Theorem~\ref{th:linial_lb}.
\end{proof}

The following known fact gives a generic lower bound on the bounded-error private-coin communication complexity in terms of the deterministic communication complexity.
\begin{fact}[\cite{KN97}, Theorem $3.14$]\label{fact:priv_det_NK}
	For every communication function $f$ and every non-negative $\epsilon$ that is bounded below $1/2$, we have that $\PrivCC_{\epsilon}(f) = \Omega(\log(\CC(f)))$.
\end{fact}

Lemma~\ref{le:1_way_priv_G_lb} now follows by combining Lemma~\ref{le:ow_det_lb_Gmt} and Fact~\ref{fact:priv_det_NK}.

\subsection{Proof of Part~\ref{part:det_unc_lb} of Theorem~\ref{th:det_sep}}\label{subsec:pf_part_iv_det_sep}

In this section, we prove Part~\ref{part:det_unc_lb} of Theorem~\ref{th:det_sep}. We assume for the sake of contradiction that there exists a $1$-way private-coin protocol $\Pi$ computing $\mathcal{F}_{\delta}$ w.r.t. the uniform distribution on $\{0,1\}^{2n}$ with error at most $\epsilon/2-2\delta-\eta$ and with communication cost $|\Pi| = o(\eta^{2} \cdot \log^{(t)}(n))$ for some positive integer $t = \Theta((\epsilon/\delta)^2)$ that will be exactly specified later on. We will use $\Pi$ to give a $1$-way private-coin protocol $\Pi'$ solving the shift communication game $\mathcal{G}_{m,t}$ with high constant probability and with communication cost $|\Pi'| \le O(\eta^{-2} \cdot |\Pi|) = o(\log^{(t)}(n))$, which would contradict Lemma~\ref{le:1_way_priv_G_lb}.

\subsubsection{Description of  Protocol $\Pi'$}
The operation of protocol $\Pi'$, which uses $\Pi$ as a black-box, is described in Protocol~\ref{alg:red_prot}. Note that the parameters in Protocol~\ref{alg:red_prot} are defined in terms of $\epsilon$ (which is a number in $[\delta,0.5]$ that is given in the statement of Theorem~\ref{th:det_sep}) and $\delta'$ (which, as mentioned above, is set to $\alpha \cdot \delta^2$ for a sufficiently small positive absolute constant $\alpha$). Also as above, $\ell$ is set to a sufficiently large function of $\delta$. Note that $\epsilon' = O(\epsilon^2)$, and hence $t = O((\epsilon/\delta)^2)$. In Protocol~\ref{alg:red_prot}, Bob is given as input a sorted tuple $\sigma$, and Alice is given as input either the prefix $\phi$ of $\sigma$ or its suffix $\psi$. In steps~\ref{st:step_1_red_prot} and \ref{st:step_2_red_prot}, Alice and Bob ``stretch'' their tuples, which amounts to each of them repeating each bit of the corresponding $0/1$ indicator vector a certain number of times and appending a certain number of zeros (see Definition~\ref{def:stretch} and Figure~\ref{fig:stretch_op} in Section~\ref{sec:app_stretch} below for a thorough definition of $\mathsf{stretch}_{r,a}$). The aim of steps~\ref{st:step_1_red_prot} and \ref{st:step_2_red_prot} is for Alice and Bob to produce a pair of subsets $(S,T)$ of $[n]$ such that $(f_S,f_T) \in \mathcal{F}_{\delta}$. The goal of Protocol~\ref{alg:red_prot} is for Bob to determine if Alice was given the prefix or the suffix of his tuple $\sigma$. To do so, Alice and Bob sample (using their private coins) $k$ random inputs $(X^{(i)}, Y^{(i)})_{i \in [k]}$ (steps~\ref{st:step_3_red_prot} and \ref{st:step_4_red_prot}) and then simulate the given private-coin protocol $\Pi$ to compute the function $f_T(X,Y)$ on the $k$ random input-pairs that were sampled (steps~\ref{st:step_6_red_prot} and \ref{st:step_7_red_prot}). Moreover, Alice sends to Bob the influential bits of $X^{(i)}$ for each $i \in [k]$ (step~\ref{st:step_8_red_prot}). The main idea will be to be for Bob to compute the empirical error corresponding to each of the prefix (step~\ref{st:step_11_red_prot}) and suffix (step~\ref{st:step_12_red_prot}), and then output the hypothesis with the smallest empirical error (step~\ref{st:step_13_red_prot}).

\begin{algorithm}
\caption{Reduction Protocol $\Pi'$}
\label{alg:red_prot}
{\bf Parameters.} $\epsilon' = 1-\cos(\epsilon \pi)$, $t = \lceil \epsilon'/\delta' \rceil$, $r = \delta' \cdot \ell$, $a = \ell \cdot (1- t \cdot \delta')$, $s = \ell \cdot (1-\delta')$, $k = \Theta(1/\eta^2)$.\\ 
{\bf Inputs.} Bob is given a sorted tuple $\sigma = (\sigma_1,\dots,\sigma_t)$ of integers with $1 \le \sigma_1 < \dots < \sigma_t \le (n-a)/r$. Alice is given a sorted tuple $\lambda \in \{\phi, \psi\}$ where $\phi \triangleq (\sigma_1,\dots,\sigma_{t-1})$ and $\psi \triangleq (\sigma_2,\dots,\sigma_t)$. \\ 
\begin{algorithmic}[1]
\State\label{st:step_1_red_prot}Alice sets $(\Lambda,S) \gets \mathsf{stretch}_{r,a}(\lambda)$.
\State\label{st:step_2_red_prot}Bob sets $(\Sigma,T) \gets \mathsf{stretch}_{r,a}(\sigma)$.
\State\label{st:step_3_red_prot} Alice uses her private randomness to sample $k$ i.i.d. strings $X^{(1)}, \dots, X^{(k)} \in_R \{0,1\}^n$.
\State\label{st:step_4_red_prot}  Bob uses his private randomness to sample $k$ i.i.d. strings $Y^{(1)}, \dots, Y^{(k)} \in_R \{0,1\}^n$.
\For{$i=1,\ldots,k$}
\State\label{st:step_6_red_prot}Alice and Bob simulate the protocol $\Pi$ on inputs $((S,X^{(i)}),(T,Y^{(i)}))$.
\State\label{st:step_7_red_prot}Bob computes the resulting output bit $B_i$.
\State\label{st:step_8_red_prot}Alice sends to Bob the sequence of bits $(X^{(i)}_{\Lambda_1},X^{(i)}_{\Lambda_2},\dots, X^{(i)}_{\Lambda_{s}})$ .
\EndFor
\State Bob sets $\Phi \gets \mathsf{stretch}_{r,a}(\phi)$ and $\Psi \gets \mathsf{stretch}_{r,a}(\psi)$.
\State\label{st:step_11_red_prot}Bob computes the ``prefix error'' $\err^{(p)} \triangleq \displaystyle\sum\limits_{i \in [k]} \1\bigg[\Sign\big(\displaystyle\sum\limits_{j \in [s]} X^{(i)}_{\Lambda_j} Y^{(i)}_{\Phi _j}\big) \neq B_i\bigg]$.
\State\label{st:step_12_red_prot}Bob computes the ``suffix error'' $\err^{(s)} \triangleq \displaystyle\sum\limits_{i \in [k]} \1\bigg[\Sign\big(\displaystyle\sum\limits_{j \in [s]} X^{(i)}_{\Lambda_j} Y^{(i)}_{\Psi _j}\big) \neq B_i\bigg]$.
\State\label{st:step_13_red_prot}Bob returns YES if $\err^{(p)} \le \err^{(s)}$ and NO otherwise.
\end{algorithmic}
\end{algorithm}

\subsubsection{The $\mathsf{stretch}_{r,a}$ Procedure}\label{sec:app_stretch}
In this section, we thoroughly define and illustrate the stretching procedure used in Protocol~\ref{alg:red_prot} and mentioned in Section~\ref{sec:det_sep}.

\begin{definition}[The $\mathsf{stretch}_{r,a}$ procedure]\label{def:stretch}
	Let $d$, $t$, $r$ and $a$ be positive integers with $t \le d$. For any sorted tuple $\sigma = (\sigma_1,\dots,\sigma_t)$ of distinct integers with $1 \le \sigma_1 < \dots < \sigma_t \le d$, we first let $z \in \{0,1\}^d$ be the $0/1$ indicator vector of the subset of $[d]$ corresponding to $\sigma$. Let $w \in \{0,1\}^{d \cdot r + a}$ be the string obtained from $z$ by repeating each of its coordinates $r$ times, and then appending $a$ ones. Namely, for each $i \in [d]$ and $j \in [r]$, we set $w_{(i-1)\cdot r + j} = z_i$ and for each $j \in [a]$, we set $w_{d \cdot r+j} = 1$. Then, we let $1 \le \Sigma_1 < \Sigma_2 < \dots < \Sigma_{t \cdot r+a} \le d \cdot r +a$ be the indices of the coordinates of $w$ that are equal to $1$. The output of $\mathsf{stretch}_{r,a}(\sigma)$ is then the pair $(\Sigma, W)$ where $\Sigma \triangleq (\Sigma_1, \Sigma_2, \dots,  \Sigma_{t \cdot r+a})$ and $W \subseteq [d\cdot r+a]$ is the support of $w$.
\end{definition}

The operation of $\mathsf{stretch}_{r,a}$ is illustrated in Figure~\ref{fig:stretch_op} in the particular case where $r = 2$, $a = 3$ and $d = 9$. The ``appending parameter'' $a$ allows us to control how small is the \emph{normalized} Hamming distance between the binary strings corresponding to $\Sigma$ and $\Phi$ (respectively $\Psi$). The purpose of the ``repetition parameter'' $r$ is the following. Consider the tuples $\phi = (2,4,5,7)$ and $\psi  = (4,5,7,9)$ in Figure~\ref{fig:stretch_op}. The number of differing coordinates between the tuples $\phi$ and $\psi$ is $4$. After stretching, the number of differing coordinates between the resulting tuples $\Phi$ and $\Psi$ is amplified to $4 \cdot r = 8$. These two notions of distance (the number of coordinates on which the tuples differ and the normalized Hamming distance between the corresponding binary strings) are important to us. The fact that these two distances are important to us is the reason why we let the $\mathsf{stretch}_{r,a}$ procedure have two equivalent outputs (a subset $W$ and a tuple $\Sigma$).

\begin{figure}[H]
\centering
\begin{tikzpicture}[scale=1.5]
\def \n {13}
\def \t {3}
\def \radius {1.15cm}
\def \margin {3} 
\def \margintwo {0.07} 
\def \radiussmall {0.7cm}
\def \radiuslarge {1.6cm}

\node[fill=none, scale=0.9, black] (n4) at (-2.75,-0.425) {Bob's subset:};

\node[fill=none, scale=0.9, black] (n4) at (-1,-0.725) {$\sigma = (2,4,5,7,9)$};

\node[fill=none, scale=0.9, color=black] (n4) at (-2.1+0.3,-0.425) {$0$};
\node[fill=none, scale=0.9, color=blue] (n4) at (-2.1+0.3+0.2,-0.425) {$1$};
\node[fill=none, scale=0.9, color=black] (n4) at (-2.1+0.3+0.4,-0.425) {$0$};

\node[fill=none, scale=0.9, color=black] (n4) at (-2.1+0.3+0.6,-0.425) {$1$};
\node[fill=none, scale=0.9, color=black] (n4) at (-2.1+0.3+0.8,-0.425) {$1$};
\node[fill=none, scale=0.9, color=black] (n4) at (-2.1+0.3+1,-0.425) {$0$};

\node[fill=none, scale=0.9, color=black] (n4) at (-2.1+0.3+1.2,-0.425) {$1$};
\node[fill=none, scale=0.9, color=black] (n4) at (-2.1+0.3+1.4,-0.425) {$0$};
\node[fill=none, scale=0.9, color=blue] (n4) at (-2.1+0.3+1.6,-0.425) {$1$};

\node[fill=none, scale=0.9, black] (n4) at (-2.75,-1.225) {Prefix of $\sigma$:};

\node[fill=none, scale=0.9, black] (n4) at (-1,-1.525) {$\phi = (2,4,5,7)$};

\node[fill=none, scale=0.9, color=black] (n4) at (-2.1+0.3,-1.225) {$0$};
\node[fill=none, scale=0.9, color=blue] (n4) at (-2.1+0.3+0.2,-1.225) {$1$};
\node[fill=none, scale=0.9, color=black] (n4) at (-2.1+0.3+0.4,-1.225) {$0$};

\node[fill=none, scale=0.9, color=black] (n4) at (-2.1+0.3+0.6,-1.225) {$1$};
\node[fill=none, scale=0.9, color=black] (n4) at (-2.1+0.3+0.8,-1.225) {$1$};
\node[fill=none, scale=0.9, color=black] (n4) at (-2.1+0.3+1,-1.225) {$0$};

\node[fill=none, scale=0.9, color=black] (n4) at (-2.1+0.3+1.2,-1.225) {$1$};
\node[fill=none, scale=0.9, color=black] (n4) at (-2.1+0.3+1.4,-1.225) {$0$};
\node[fill=none, scale=0.9, color=blue] (n4) at (-2.1+0.3+1.6,-1.225) {$0$};

\node[fill=none, scale=0.9, black] (n4) at (-2.75,-1.9) {Suffix of $\sigma$:};
\node[fill=none, scale=0.9, black] (n4) at (-1,-2.2) {$\psi = (4,5,7,9)$};

\node[fill=none, scale=0.9, color=black] (n4) at (-2.1+0.3,-1.9) {$0$};
\node[fill=none, scale=0.9, color=blue] (n4) at (-2.1+0.3+0.2,-1.9) {$0$};
\node[fill=none, scale=0.9, color=black] (n4) at (-2.1+0.3+0.4,-1.9) {$0$};

\node[fill=none, scale=0.9, color=black] (n4) at (-2.1+0.3+0.6,-1.9) {$1$};
\node[fill=none, scale=0.9, color=black] (n4) at (-2.1+0.3+0.8,-1.9) {$1$};
\node[fill=none, scale=0.9, color=black] (n4) at (-2.1+0.3+1,-1.9) {$0$};

\node[fill=none, scale=0.9, color=black] (n4) at (-2.1+0.3+1.2,-1.9) {$1$};
\node[fill=none, scale=0.9, color=black] (n4) at (-2.1+0.3+1.4,-1.9) {$0$};
\node[fill=none, scale=0.9, color=blue] (n4) at (-2.1+0.3+1.6,-1.9) {$1$};

(-0.1,-1.9) -- (-0.1,-0.5) node [black,midway,xshift=0.8cm]{\footnotesize};

\draw[->,color=black] (0,-0.425) -- (1.35,-0.425);

\node[fill=none, scale=0.8, color=purple!75!black] (n4) at (0.65,-0.3) {$\mathsf{stretch}_{r,a}$};

\node[fill=none, scale=0.9, black] (n4) at (3.65,-0.725) {$\Sigma = (3,4,7,8,9,10,13,14,17,18,19,20,21)$};

\node[fill=none, scale=0.9, color=black] (n4) at (1.25+0.4,-0.425) {$0$};
\node[fill=none, scale=0.9, color=black] (n4) at (1.25+0.4+0.2,-0.425) {$0$};

\node[fill=none, scale=0.9, color=blue] (n4) at (1.25+0.4+0.4,-0.425) {$1$};
\node[fill=none, scale=0.9, color=blue] (n4) at (1.25+0.4+0.6,-0.425) {$1$};

\node[fill=none, scale=0.9, color=black] (n4) at (1.25+0.4+0.8,-0.425) {$0$};
\node[fill=none, scale=0.9, color=black] (n4) at (1.25+0.4+1,-0.425) {$0$};

\node[fill=none, scale=0.9, color=black] (n4) at (1.25+0.4+1.2,-0.425) {$1$};
\node[fill=none, scale=0.9, color=black] (n4) at (1.25+0.4+1.4,-0.425) {$1$};

\node[fill=none, scale=0.9, color=black] (n4) at (1.25+0.4+1.6,-0.425) {$1$};
\node[fill=none, scale=0.9, color=black] (n4) at (1.25+0.4+1.8,-0.425) {$1$};

\node[fill=none, scale=0.9, color=black] (n4) at (1.25+0.4+2,-0.425) {$0$};
\node[fill=none, scale=0.9, color=black] (n4) at (1.25+0.4+2.2,-0.425) {$0$};

\node[fill=none, scale=0.9, color=black] (n4) at (1.25+0.4+2.4,-0.425) {$1$};
\node[fill=none, scale=0.9, color=black] (n4) at (1.25+0.4+2.6,-0.425) {$1$};

\node[fill=none, scale=0.9, color=black] (n4) at (1.25+0.4+2.8,-0.425) {$0$};
\node[fill=none, scale=0.9, color=black] (n4) at (1.25+0.4+3,-0.425) {$0$};

\node[fill=none, scale=0.9, color=blue] (n4) at (1.25+0.4+3.2,-0.425) {$1$};
\node[fill=none, scale=0.9, color=blue] (n4) at (1.25+0.4+3.4,-0.425) {$1$};

\node[fill=none, scale=0.9, color=black] (n4) at (1.25+0.4+3.6,-0.425) {$1$};
\node[fill=none, scale=0.9, color=black] (n4) at (1.25+0.4+3.8,-0.425) {$1$};
\node[fill=none, scale=0.9, color=black] (n4) at (1.25+0.4+4,-0.425) {$1$};

\draw[->,color=black] (0,-1.225) -- (1.35,-1.225);

\node[fill=none, scale=0.8, color=purple!75!black] (n4) at (0.65,-1.1) {$\mathsf{stretch}_{r,a}$};

\node[fill=none, scale=0.9, black] (n4) at (3.65,-1.525) {$\Phi = (3,4,7,8,9,10,13,14,19,20,21)$};

\node[fill=none, scale=0.9, color=black] (n4) at (1.25+0.4,-1.225) {$0$};
\node[fill=none, scale=0.9, color=black] (n4) at (1.25+0.4+0.2,-1.225) {$0$};

\node[fill=none, scale=0.9, color=blue] (n4) at (1.25+0.4+0.4,-1.225) {$1$};
\node[fill=none, scale=0.9, color=blue] (n4) at (1.25+0.4+0.6,-1.225) {$1$};

\node[fill=none, scale=0.9, color=black] (n4) at (1.25+0.4+0.8,-1.225) {$0$};
\node[fill=none, scale=0.9, color=black] (n4) at (1.25+0.4+1,-1.225) {$0$};

\node[fill=none, scale=0.9, color=black] (n4) at (1.25+0.4+1.2,-1.225) {$1$};
\node[fill=none, scale=0.9, color=black] (n4) at (1.25+0.4+1.4,-1.225) {$1$};

\node[fill=none, scale=0.9, color=black] (n4) at (1.25+0.4+1.6,-1.225) {$1$};
\node[fill=none, scale=0.9, color=black] (n4) at (1.25+0.4+1.8,-1.225) {$1$};

\node[fill=none, scale=0.9, color=black] (n4) at (1.25+0.4+2,-1.225) {$0$};
\node[fill=none, scale=0.9, color=black] (n4) at (1.25+0.4+2.2,-1.225) {$0$};

\node[fill=none, scale=0.9, color=black] (n4) at (1.25+0.4+2.4,-1.225) {$1$};
\node[fill=none, scale=0.9, color=black] (n4) at (1.25+0.4+2.6,-1.225) {$1$};

\node[fill=none, scale=0.9, color=black] (n4) at (1.25+0.4+2.8,-1.225) {$0$};
\node[fill=none, scale=0.9, color=black] (n4) at (1.25+0.4+3,-1.225) {$0$};

\node[fill=none, scale=0.9, color=blue] (n4) at (1.25+0.4+3.2,-1.225) {$0$};
\node[fill=none, scale=0.9, color=blue] (n4) at (1.25+0.4+3.4,-1.225) {$0$};

\node[fill=none, scale=0.9, color=black] (n4) at (1.25+0.4+3.6,-1.225) {$1$};
\node[fill=none, scale=0.9, color=black] (n4) at (1.25+0.4+3.8,-1.225) {$1$};
\node[fill=none, scale=0.9, color=black] (n4) at (1.25+0.4+4,-1.225) {$1$};

\draw[->,color=black] (0,-1.9) -- (1.35,-1.9);

\node[fill=none, scale=0.8, color=purple!75!black] (n4) at (0.65,-1.7275) {$\mathsf{stretch}_{r,a}$};

\node[fill=none, scale=0.9, black] (n4) at (3.65,-2.2) {$\Psi = (7,8,9,10,13,14,17,18,19,20,21)$};

\node[fill=none, scale=0.9, color=black] (n4) at (1.25+0.4,-1.9) {$0$};
\node[fill=none, scale=0.9, color=black] (n4) at (1.25+0.4+0.2,-1.9) {$0$};

\node[fill=none, scale=0.9, color=blue] (n4) at (1.25+0.4+0.4,-1.9) {$0$};
\node[fill=none, scale=0.9, color=blue] (n4) at (1.25+0.4+0.6,-1.9) {$0$};

\node[fill=none, scale=0.9, color=black] (n4) at (1.25+0.4+0.8,-1.9) {$0$};
\node[fill=none, scale=0.9, color=black] (n4) at (1.25+0.4+1,-1.9) {$0$};

\node[fill=none, scale=0.9, color=black] (n4) at (1.25+0.4+1.2,-1.9) {$1$};
\node[fill=none, scale=0.9, color=black] (n4) at (1.25+0.4+1.4,-1.9) {$1$};

\node[fill=none, scale=0.9, color=black] (n4) at (1.25+0.4+1.6,-1.9) {$1$};
\node[fill=none, scale=0.9, color=black] (n4) at (1.25+0.4+1.8,-1.9) {$1$};

\node[fill=none, scale=0.9, color=black] (n4) at (1.25+0.4+2,-1.9) {$0$};
\node[fill=none, scale=0.9, color=black] (n4) at (1.25+0.4+2.2,-1.9) {$0$};

\node[fill=none, scale=0.9, color=black] (n4) at (1.25+0.4+2.4,-1.9) {$1$};
\node[fill=none, scale=0.9, color=black] (n4) at (1.25+0.4+2.6,-1.9) {$1$};

\node[fill=none, scale=0.9, color=black] (n4) at (1.25+0.4+2.8,-1.9) {$0$};
\node[fill=none, scale=0.9, color=black] (n4) at (1.25+0.4+3,-1.9) {$0$};

\node[fill=none, scale=0.9, color=blue] (n4) at (1.25+0.4+3.2,-1.9) {$1$};
\node[fill=none, scale=0.9, color=blue] (n4) at (1.25+0.4+3.4,-1.9) {$1$};

\node[fill=none, scale=0.9, color=black] (n4) at (1.25+0.4+3.6,-1.9) {$1$};
\node[fill=none, scale=0.9, color=black] (n4) at (1.25+0.4+3.8,-1.9) {$1$};
\node[fill=none, scale=0.9, color=black] (n4) at (1.25+0.4+4,-1.9) {$1$};

\end{tikzpicture}
\caption{Operation of the $\mathsf{stretch}_{r,a}$ procedure for $r=2$ and $a=3$.}\label{fig:stretch_op}
\end{figure}

\subsubsection{Analysis of Protocol $\Pi'$}
We now turn to the formal analysis of Protocol~\ref{alg:red_prot}. First, note that the communication cost of protocol $\Pi'$ satisfies $|\Pi'| \le O(\eta^{-2} \cdot |\Pi|)$. Let $\Lambda$, $\Sigma$, $\Phi$ and $\Psi$ be the ordered sequences defined in the operation of Protocol~\ref{alg:red_prot}. Define the functions $g$ and $h$ as
\begin{subequations}
	\begin{empheq}{align}
	g(X,Y) &\triangleq \Sign\big(\displaystyle\sum\limits_{j \in [s]} X_{\Lambda_j} Y_{\Phi _j}\big),\label{eq:g_def}\\
	h(X,Y) &\triangleq \Sign\big(\displaystyle\sum\limits_{j \in [s]} X_{\Lambda_j} Y_{\Psi _j}\big).\label{eq:h_def}
	\end{empheq}
\end{subequations}

Note that steps~\ref{st:step_11_red_prot} and \ref{st:step_12_red_prot} of the protocol compute the empirical errors of functions $g$ and $h$ respectively. Since $\Lambda \in \{\Phi, \Psi\}$, let's assume WLOG that $\Lambda = \Phi$ and show that the protocol $\Pi'$ returns YES with high probability. The case where $\Lambda = \Psi$ is symmetric. When $\Lambda = \Phi$, we have that $g = f_S$ where $S$ is the subset of $[n]$ that Alice gets in step~\ref{st:step_1_red_prot}. The operation of the $\mathsf{stretch}_{r,a}$ procedure (Definition~\ref{def:stretch}) guarantees that $S \subseteq [n]$, $|T| = \ell$ and $|T \setminus S| = \delta' \cdot \ell$. Part~\ref{part:det_dist} of Theorem~\ref{th:det_sep} then implies that $\Delta(g,f_T) \le \delta$. In order to show that in step~\ref{st:step_13_red_prot} Bob returns YES with high probability, the main idea will be to lower bound the distance between the functions $g$ and $h$. This is done in the next lemma. 

\begin{lem}\label{le:dist_sep}
	The functions $g$ and $h$ defined in Equations~(\ref{eq:g_def}) and (\ref{eq:h_def}) satisfy $\Delta(g,h) \geq \epsilon - \delta$.
\end{lem}

To prove Lemma~\ref{le:dist_sep}, we use the next lemma which spells out the distribution of the sequence $(X_{\Lambda_j} Y_{\Phi _j},X_{\Lambda_j} Y_{\Psi _j})_{j \in [s]}$ of random variables.

\begin{lem}\label{le:indep_coord}
	The random variables $(X_{\Lambda_1} Y_{\Phi _1},X_{\Lambda_1} Y_{\Psi _1}), (X_{\Lambda_2} Y_{\Phi _2},X_{\Lambda_2} Y_{\Psi _2}), \dots, (X_{\Lambda_s} Y_{\Phi _s},X_{\Lambda_s} Y_{\Psi _s})$ are independent, and they are distributed as follows:
	\begin{enumerate}
		\item\label{part:indep_1} For $1 \le j \le s-a$, $(X_{\Lambda_j} Y_{\Phi _j},X_{\Lambda_j} Y_{\Psi _j})$ is uniformly distributed on $\{\pm 1\}^2$.
		\item\label{part:indep_2} For $s-a < j \le s$, $(X_{\Lambda_j} Y_{\Phi _j},X_{\Lambda_j} Y_{\Psi _j})$ is uniformly distributed on $\{ (-1,-1), (+1,+1)\}$.
	\end{enumerate}
\end{lem}

\begin{proof}[Proof of Lemma~\ref{le:indep_coord}]
	By the operation of the $\mathsf{stretch}_{r,a}$ procedure (Definition~\ref{def:stretch}) in steps~\ref{st:step_1_red_prot} and \ref{st:step_2_red_prot} of Protocol~\ref{alg:red_prot}, for every $s-a < j \le s$, it holds that $(X_{\Lambda_j} Y_{\Phi _j},X_{\Lambda_j} Y_{\Psi _j}) = (X_{j} Y_{j},X_{j} Y_{j})$ and that $X_j$ and $Y_j$ do not contribute to any other $(X_{\Lambda_{j'}} Y_{\Phi_{j'}},X_{\Lambda_{j'}} Y_{\Psi_{j'}})$ pair. This implies part~\ref{part:indep_1} and that for each $s-a < j \le s$, the pair $(X_{\Lambda_j} Y_{\Phi _j},X_{\Lambda_j} Y_{\Psi _j})$ is independent of all other pairs in the sequence.
	
	We next prove by induction on $1 \le j \le s-a$ that conditioned on $(X_{\Lambda_{<j}} Y_{\Phi_{<j}},X_{\Lambda_{<j}} Y_{\Psi _{<j}})$ taking any particular value, $(X_{\Lambda_j} Y_{\Phi _j},X_{\Lambda_j} Y_{\Psi _j})$ is uniformly distributed on $\{\pm 1\}^2$. To see this (assuming WLOG that $\Lambda = \Phi$), note that $(X_{\Lambda_j} Y_{\Phi _j},X_{\Lambda_j} Y_{\Psi _j}) = (X_{\Lambda_j} Y_{\Lambda_j},X_{\Lambda_j} Y_{\Psi_j})$ where $X_{\Lambda_j}$ and $Y_{\Psi _j}$ do not appear in $(X_{\Lambda_{<j}} Y_{\Phi_{<j}},X_{\Lambda_{<j}} Y_{\Psi _{<j}})$. Hence, conditioned on $(X_{\Lambda_{<j}} Y_{\Phi_{<j}},X_{\Lambda_{<j}} Y_{\Psi _{<j}})$ taking any particular value, $X_{\Lambda_j} Y_{\Phi _j}$ is a uniformly random bit (because of $X_{\Lambda_j}$). Moreover, conditioned on $(X_{\Lambda_{<j}} Y_{\Phi_{<j}},X_{\Lambda_{<j}} Y_{\Psi _{<j}})$ and $X_{\Lambda_j} Y_{\Phi _j}$ taking any particular values, $X_{\Lambda_j} Y_{\Psi_j}$ is a uniformly random bit (because of $Y_{\Psi _j}$). This completes the proof of the lemma.
\end{proof}

We now use Lemma~\ref{le:indep_coord} along with a two-dimensional Central Limit Theorem in order to prove Lemma~\ref{le:dist_sep}.

\begin{proof}[Proof of Lemma~\ref{le:dist_sep}]
	By Lemma~\ref{le:indep_coord}, independently for each $1 \le j \le s-a$, $(X_{\Lambda_j} Y_{\Phi _j},X_{\Lambda_j} Y_{\Psi _j})$ is a $\rho_j$-correlated random pair with $\rho_j = 0$, and independently for each $s-a < j \le s$, $(X_{\Lambda_j} Y_{\Phi _j},X_{\Lambda_j} Y_{\Psi _j})$ is a $\rho_j$-correlated random pair with $\rho_j = 1$. Hence, the average correlation across coordinates is
	\begin{equation*}
	\rho \triangleq k^{-1} \cdot \displaystyle\sum\limits_{j \in [k]} \rho_j = \frac{a}{s} = \frac{\ell \cdot (1- t \cdot \delta')}{\ell \cdot (1-\delta')} = 1-\frac{(t-1) \cdot \delta'}{1-\delta'} \le 1-(\epsilon'-\delta'),
	\end{equation*}
	where the last inequality used the setting of $t = \lceil \epsilon'/\delta' \rceil$ in Protocol~\ref{alg:red_prot}. Denoting by $\Sigma_j$ the covariance matrix of $(X_{\Lambda_j} Y_{\Phi _j},X_{\Lambda_j} Y_{\Psi _j})$ for every $j \in [s]$, the average (across coordinates) covariance matrix is then given by $\Sigma \triangleq k^{-1} \cdot \displaystyle\sum\limits_{j \in [k]} \Sigma_j = \begin{bmatrix}
	1       & \rho \\
	\rho       & 1 
	\end{bmatrix}$. Note that the smallest eigenvalue of $\Sigma$ is $\lambda = 1-|\rho|$. Let $(X',Y')$ be a pair of zero-mean $\rho$-correlated Gaussians. By the two-dimensional Berry-Esseen Theorem~\ref{th:2d_berry_ess_bin_not_id} and Sheppard's formula (i.e., Fact~\ref{fact:sheppard}), we get that
	\begin{align*}
	\Delta(g,h) &= \Pr[ g(X,Y) \neq h(X,Y)]\\ 
	&= \Pr[ g(X,Y) = 0, h(X,Y) = 1] + \Pr[ g(X,Y) = 1, h(X,Y) = 0]\\ 
	&= \Pr[ X' <0, Y \geq 0] \pm O\bigg(\frac{1}{(1-|\rho|)^{3/2} \cdot \sqrt{s}}\bigg) + \Pr[ X' \geq 0, Y < 0] \pm O\bigg(\frac{1}{(1-|\rho|)^{3/2} \cdot \sqrt{s}}\bigg)\\ 
	&= \Pr[ \Sign(X') \neq \Sign(Y')] \pm O\bigg(\frac{1}{(1-|\rho|)^{3/2} \cdot \sqrt{\ell}}\bigg)\\ 
	&= \frac{\arccos(\rho)}{\pi} \pm O\bigg(\frac{1}{(1-|\rho|)^{3/2} \cdot \sqrt{\ell}}\bigg)\\ 
	&\geq \frac{\arccos(1-(\epsilon'-\delta'))}{\pi} - O\bigg(\frac{1}{(1-|\rho|)^{3/2} \cdot \sqrt{\ell}}\bigg)\\ 
	&\geq \frac{\arccos(1-\epsilon')}{\pi} - O(\sqrt{\delta'})- O\bigg(\frac{1}{(1-|\rho|)^{3/2} \cdot \sqrt{\ell}}\bigg)\\ 
	&\geq \epsilon - \delta,
	\end{align*}
	where the last inequality follows by setting $\epsilon' = 1-\cos(\epsilon \pi)$ for the given $\epsilon \in [\delta, 0.5]$, setting $\ell$ to be a sufficiently large function of $\delta$, and setting $\delta' = \alpha \cdot \delta^2$ for a sufficiently small positive absolute constant $\alpha$.
\end{proof}

We are now ready to complete the proof of Part~\ref{part:det_unc_lb} of Theorem~\ref{th:det_sep}. In Protocol~\ref{alg:red_prot}, the tuples $((X^{(i)}, Y^{(i)}), B_i)_{i \in [k]}$ are i.i.d. samples corresponding to a function $q$ which, by the error guarantee of protocol $\Pi$, is $(\epsilon/2-2\delta-\eta)$-close to $f_T$. Since $\Delta(g,f_T) \le \delta$, we get that $\Delta(g,q) \le (\epsilon/2-\delta-\eta)$. By Lemma~\ref{le:dist_sep}, we also get that $\Delta(h,q) \geq (\epsilon/2+\eta)$. By Hoeffding's bound (i.e., Fact~\ref{fact:hoeffding}), for $k = \Theta(1/\eta^2)$, the empirical error $\eps^{(p)}$ of $g$ on the samples $((X^{(i)}, Y^{(i)}), B_i)_{i \in [k]}$ is less than the empirical error $\eps^{(s)}$ of $h$ on these samples, with high constant probability. Hence, Bob returns YES with high constant probability. A symmetric argument shows that when $\Lambda = \Psi$, Bob returns NO with high constant probability, which completes the proof.

\begin{fact}[Hoeffding's bound]\label{fact:hoeffding}
	Consider a coin that shows up head with probability $p$. Let $H(k)$ be the number of heads obtained in $k$ independent tosses of this coin. Then, for every $\epsilon >0$, $\Pr[ |H(k)-p k| > \epsilon k] \le 2 e^{-2 \epsilon^2 k}$.
\end{fact}

\begin{note}\label{rem:priv_elab}
	As mentioned in Note~\ref{rem:priv}, the above construction cannot give a separation larger than $\Theta(\log\log{n})$. This is because using private randomness, Bob can learn the set $S$ using $O(\log\log{n})$ bits of communication (see, e.g., \cite{brody2014beyond}). Additionally, Alice can send the coordinates of $X$ indexed by the elements of $S$ to Bob who can then compute $f_S(X,Y)$.
\end{note}

\section{Construction for Public-Coin Uncertain Protocols}\label{sec:rand_imp_sep}
In this section, we describe the construction that is used to prove Theorem~\ref{th:imp_rand_sep}. We set
\begin{equation*}
\delta' \triangleq c \cdot \delta^2 ~~ \text{ and } ~~ \epsilon \triangleq \sqrt{\delta'/n},\tag{\dag}\label{eq:eps_prime_delt_prime}
\end{equation*}
where $\delta$ is the parameter from the statement of Theorem~\ref{th:imp_rand_sep}, and $c >0$ is a small-enough absolute constant. To define our input distribution, we first define a slightly more general distribution $\mu_{\eta}$. The support of $\mu_{\eta}$ is $\{0,1\}^{k n} \times \{0,1\}^{k n}$ and we will view the coordinates of a sample $(x,y) \sim \mu_{\eta}$ as $x = (x^{(i)})_{i \in [k]}$ and $y = (y^{(i)})_{i \in [k]}$ with $x^{(i)}, y^{(i)} \in \{0,1\}^n$ for all $i \in [k]$. A sample $(x,y) \sim \mu_{\eta}$ is generated by letting $x \in_R \{0,1\}^{k n}$ and for all $i \in [k]$ and $j \in [n]$, independently setting $y^{(i)}_j$ to be an $\eta$-noisy copy of $x^{(i)}_j$. In other words, we set $y^{(i)}_j = x^{(i)}_j$ w.p. $1-\eta$ and $y^{(i)}_j = 1-x^{(i)}_j$ w.p. $\eta$. Our input distribution is then $\mu_{2\epsilon -2\epsilon^2}$.

We now define our function class $\mathcal{F}_{\epsilon}$. Each function in our universe is specified by a sequence of subsets $S \triangleq (S^{(i)} \subseteq [n])_{i \in [k]}$ and it is of the form $f_S:\{0,1\}^{2 \cdot k \cdot n} \to \{0,1\}$ with\footnote{We will use the symbol $S^{(i)}$ to denote both the subset of $[n]$ and its corresponding $0/1$ indicator vector.}
\begin{equation}\label{eq:rand_univ_func}
f_S(x,y) \triangleq \Sign\big( \displaystyle\sum\limits_{i \in [k]} (-1)^{\langle S^{(i)}, x^{(i)} \oplus y^{(i)} \rangle}\big)
\end{equation}
for all $x,y \in \{0,1\}^{k \cdot n}$, where in Eq.~(\ref{eq:rand_univ_func}) the inner product is over $\mathbb{F}_2$, the sum is over $\mathbb{R}$ and $x^{(i)} \oplus y^{(i)}$ denotes the coordinate-wise XOR of the two length-$n$ binary strings $x^{(i)}$ and $y^{(i)}$. The function class is then defined by $\mathcal{F}_{\epsilon} \triangleq \{ (f_S, f_T): |S^{(i)} \triangle T^{(i)}| \le \epsilon \cdot n \text{ for all } i \in [k]\}$.

We now give the proof of Part~\ref{part:rand_dist} of Theorem~\ref{th:imp_rand_sep}. It follows from known bounds on the noise stability of the majority function.

\begin{proof}[Proof of Part~\ref{part:rand_dist} of Theorem~\ref{th:imp_rand_sep}]
	Let $(f_S,f_T) \in \mathcal{F}_{\epsilon}$, and denote $a_i \triangleq (-1)^{\langle S^{(i)}, x^{(i)} \oplus y^{(i)} \rangle}$ and $b_i \triangleq (-1)^{\langle T^{(i)}, x^{(i)} \oplus y^{(i)} \rangle}$ for every $i \in [k]$. Also, let $a \triangleq (a_i)_{i \in [k]}$ and $b \triangleq (b_i)_{i \in [k]}$. Note that $(a_i,b_i)$ is a pair of $\rho_i$-correlated random strings with $\rho_i \geq (1-2\delta')$. Since $\Delta_{\mu_{2\epsilon -2\epsilon^2}}(f_S,f_T)$ increases when $\rho_i$ decreases, we assume WLOG that $\rho_i = 1-2\delta' \triangleq \rho$ for all $i \in [k]$.
	
	Recall that the \emph{noise stability} of a function $h:\{0,1\}^k \to \{\pm 1\}$ is defined as $\mathsf{Stab}_{\rho}(h) = \Ex[h(x) h(y)]$ where $(x,y)$ is a random pair of $\rho$-correlated strings. Let the function $\widetilde{\Maj}_k$ be defined by $\widetilde{\Maj}_k(x) = (-1)^{\Maj_k(x)}$ for all $x \in \{0,1\}^k$. Recall (see~\cite{ODonnellBook}) that the noise stability of $\widetilde{\Maj}_k$ satisfies
	\begin{equation*}
	\mathsf{Stab}_{\rho}(\widetilde{\Maj}_k) \geq 1- \frac{2}{\pi} \arccos(\rho).
	\end{equation*}
	Hence, we get that
	\begin{align*}
	\Delta_{\mu_{2\epsilon -2\epsilon^2}}(f_S,f_T) &= \Pr[f_S(x,y) \neq f_T(x,y)]\\ 
	&\le \frac{1-\mathsf{Stab}_{\rho}(\widetilde{\Maj}_k)}{2}\\ 
	&\le \frac{\arccos(\rho)}{\pi}\\ 
	&= O(\sqrt{\delta'})\\ 
	&\le \delta,
	\end{align*}
	where the last equality uses the facts that $\rho = 1-2\delta'$ and that $\arccos(1-x) = O(\sqrt{x})$ for small positive values of $x$, and the last inequality follows from the setting of $\delta'$ in (\ref{eq:eps_prime_delt_prime}).
\end{proof}

We now give the (straightforward) proof of Part~\ref{part:cert_ub} of Theorem~\ref{th:imp_rand_sep}.
\begin{proof}[Proof of Part~\ref{part:cert_ub} of Theorem~\ref{th:imp_rand_sep}]
	Note that $\langle S^{(i)}, x^{(i)} \oplus y^{(i)} \rangle = \langle S^{(i)}, x^{(i)} \rangle \oplus \langle S^{(i)}, y^{(i)} \rangle$. Thus, when both Alice and Bob know $S$, Alice can send the $k$ bits $(\langle S^{(i)}, x^{(i)} \rangle)_{i \in [k]}$ to Bob who can then output the value $f_{S}(x,y)$.
\end{proof}

In order to prove Part~\ref{part:uncert_lb} of Theorem~\ref{th:imp_rand_sep}, we first define (as in \cite{ghazi2015communication}) a communication problem in the standard distributional model that reduces to solving the contextually-uncertain problem specified by the function class $\mathcal{F}_{\epsilon}$ and the distribution $\mu_{2\epsilon -2\epsilon^2}$. For distributions $\phi$ and $\psi$, we denote by $\phi \otimes \psi$ the joint distribution of a sample from $\phi$ and an independent sample from $\psi$. The new problem is defined as follows.

\paragraph*{Inputs:} Alice's input is a pair $(S,x)$ where $S \triangleq (S^{(i)} \subseteq [n])_{i \in [k]}$ and $x \in \{0,1\}^{k \cdot n}$. Bob's input is a pair $(T,y)$ where $T \triangleq (T^{(i)} \subseteq [n])_{i \in [k]}$ and $y \in \{0,1\}^{k \cdot n}$.

\paragraph*{Distribution:} Let $\mathcal{D}_q$ be the distribution on the pair $(S,T)$ of sequences of $k$ subsets of $[n]$, which is defined by independently setting, for each $i \in [k]$, $S^{(i)}$ to be a uniformly-random subset of $[n]$, and $T^{(i)}$ to be a $q$-noisy copy of $S^{(i)}$. The distribution on the inputs of Alice and Bob is then given by $\nu_{\epsilon} \triangleq \mathcal{D}_{\epsilon} \otimes \mu_{2\epsilon -2\epsilon^2}$ with $\epsilon = \sqrt{\delta'/n}$.

\paragraph*{Function:} The goal is to compute the function $F:\{0,1\}^{2kn} \times \{0,1\}^{2kn} \to \{0,1\}$ defined by $F((S,x),(T,y)) = f_T(x,y) = \Sign\big( \sum_{i \in [k]} (-1)^{\langle T^{(i)}, x^{(i)} \oplus y^{(i)} \rangle}\big)$.\\ 

The next proposition follows from a simple application of the Chernoff bound.
\begin{proposition}\label{prop:rand_reduction}
For any $\theta >0$, $\owCCU_{\theta}^{\mu_{2\epsilon-2-\epsilon^2}}(\calf_{\epsilon}) \geq \owCC^{\nu_{\epsilon}}_{\theta + \theta'}(F)$ with $\theta' = 2^{-\Theta(\epsilon \cdot n)}$.
\end{proposition}

We will prove the following lower bound on $\owCC^{\nu_{\epsilon}}_{\theta}(F)$, which along with Proposition~\ref{prop:rand_reduction} and the settings of $\epsilon$ and $\delta'$ in (\ref{eq:eps_prime_delt_prime}), implies Part~\ref{part:uncert_lb} of Theorem~\ref{th:imp_rand_sep}:
\begin{lem}\label{le:rand_lb_std}
For every sufficiently small positive constant $\theta$, $\owCC^{\nu_{\epsilon}}_{\theta}(F) = \Omega(k \cdot \epsilon \cdot n)$.
\end{lem}
We now prove Lemma~\ref{le:rand_lb_std} (which is the bulk of the proof of Theorem~\ref{th:imp_rand_sep}). Subsection~\ref{subsec:known_res_tools} summarizes some known results that we use in Subsections~\ref{subsec:sim_prot} and~\ref{subsec:pf_rand_lb_std}. In Subsection~\ref{subsec:sim_prot}, we prove a ``Simulation Lemma'' that will be useful to us. In Subsection~\ref{subsec:pf_rand_lb_std}, we prove Lemma~\ref{le:rand_lb_std}.

\subsection{Proof Preliminaries for Subsections~\ref{subsec:sim_prot} and~\ref{subsec:pf_rand_lb_std}}\label{subsec:known_res_tools}

In this subsection, we state some tools and known results that we use in the proofs in Subsections~\ref{subsec:sim_prot} and~\ref{subsec:pf_rand_lb_std}. We use the following strong direct product theorem for discrepancy of \cite{lee2008direct}\footnote{We point out that the statement of Corollary $23$ of \cite{lee2008direct} has a small inaccuracy: the additive $2^{-k \cdot (1-H_b(\tau))}$ term in Lemma~\ref{le:strong_dp} is missing. This term is clearly needed as one can always guess $f^{(k)}$ with probability $2^{-k}$. The statement that we use (Lemma~\ref{le:strong_dp}) can be obtained by combining Theorem $22$ of \cite{lee2008direct} and the proof of Proposition $1.4$ of \cite{viola2008norms}.}.
\begin{lem}[Corollary $23$ of \cite{lee2008direct}]\label{le:strong_dp}
	Let $f: X \times Y \to \{0,1\}$ be a Boolean function and $P$ a probability distribution over $X \times Y$. If $\CC^{P}_{1/2-w/2}(f) \geq C$ is proved using the discrepancy method, then the success probability under distribution $P^{\otimes k}$ of any $k C/3$ bit protocol computing the vector of solutions $f^{(k)}$ is at most $(8w)^{\tau \cdot k} + 2^{-k \cdot (1-H_b(\tau))}$ where $\tau$ is any positive constant less than $0.5$.
\end{lem}

Let $\xi_{\epsilon}$ be the distribution that is obtained by projecting $\nu_{\epsilon}$ on one of the $k$ blocks and marginalizing over the remaining $k-1$ blocks. Namely, $\nu_{\epsilon} = \xi_{\epsilon}^{\otimes k}$. Define the ``base function'' $G((\tilde{S},\tilde{x}),(\tilde{T},\tilde{y}))$ by $G((\tilde{S},\tilde{x}),(\tilde{T},\tilde{y})) = \langle \tilde{T}, \tilde{x} \oplus \tilde{y} \rangle$ for every $\tilde{S}, \tilde{T} \subseteq [n]$ and $\tilde{x}, \tilde{y} \in \{0,1\}^n$. We will use the following lower-bound on the distributional complexity of $G$ over $\xi_{\epsilon}$ that was proved in \cite{ghazi2015communication} using the discrepancy method.
\begin{lem}[\cite{ghazi2015communication}]\label{le:soda_disc_up_bd}
	For any $w = 2^{-o(\sqrt{\delta' n})}$, we have that $\CC^{\xi_{\epsilon}}_{1/2-w/2}(G) \geq \Omega(\epsilon \cdot n)$, and it is proved using the discrepancy method.
\end{lem}

Combining Lemma~\ref{le:strong_dp} and Lemma~\ref{le:soda_disc_up_bd} implies the next corollary.
\begin{corollary}\label{cor:sdp_disc_comb}
	For every positive constant $\gamma$, any deterministic protocol computing $G^{(k)}$ correctly with probability at least $(0.5+\gamma)^k$ with respect to the distribution $\xi_{\epsilon}^{\otimes k} = \nu_{\epsilon}$ should be communicating $\Omega(k \cdot \epsilon \cdot n)$ bits.
\end{corollary}

We define the Hamming distance function $\HD_k:\{0,1\}^k \times \{0,1\}^k \to \{0,1\}$ as follows. For all $x,y \in \{0,1\}^k$, $\HD_k(x,y) = 1$ if the Hamming distance between $x$ and $y$ is at least $\floor{k/2}$ and $\HD_k(x,y) = 0$ otherwise. Let $\mathcal{U}_{2 k}$ denote the uniform distribution on $\{0,1\}^{2 k}$. 
\begin{lem}[\cite{woodruff2007efficient}]\label{le:woodruff_lb}
	For every sufficiently small $\epsilon > 0$, it holds that $\owCC^{\mathcal{U}_{2 k}}_{\epsilon}(\HD_k) = \Omega(k)$.
\end{lem}

The next lemma of \cite{jain2003direct} compresses a $1$-way private-coin protocol with external information cost\footnote{The external information cost of a protocol is the amount of information that it reveals about the inputs to an external observer. For a $1$-way private-coin protocol, it is given by $I(X,Y;M)$ where $M$ is the single message sent from Alice to Bob.} $I$ into a $1$-way deterministic protocol with communication cost $O(I)$.
\begin{lem}[Result $1$ of \cite{jain2003direct}]\label{le:jrs_comp_bd_rds}
	Suppose that $\Pi$ is a $1$-way private-coin randomized protocol for $f:\mathcal{X} \times \mathcal{Y} \to \mathcal{Z}$. Let the average error of $\Pi$ under a probability distribution $\mu$ on the inputs $\mathcal{X} \times \mathcal{Y}$ be $\theta$. Let $X,Y$ denote the random variables corresponding to Alice's and Bob's inputs respectively. Let $M$ denote the single message sent by Alice to Bob. Suppose $I(X,Y; M) \le a$. Let $\zeta > 0$. Then, there is another deterministic $1$-way protocol $\Pi'$ with the following properties:
	\begin{enumerate}
		\item The communication cost of $\Pi'$ is at most $\frac{2(a+1)}{\zeta^2} + \frac{2}{\zeta}$ bits.
		\item The distributional error of $\Pi'$ under $\mu$ is at most $\theta + 2\zeta$.
	\end{enumerate}
\end{lem}

We will also use the next lemma.
\begin{lem}\label{le:estimation_k_d}
	Let $(Q,W,B)$ be correlated random variables with $Q \in \mathcal{Q}$, $W \in \mathcal{W}$ and $B \in \{0,1\}^k$. Let $\alpha \in (0,1]$ be any constant. If $I(B;W | Q) \geq \alpha \cdot k$, then there exists a positive constant $\beta$ that only depends on $\alpha$, and a deterministic function $E:\mathcal{Q} \times \mathcal{W} \to \{0,1\}^k$ such that $E(Q,W) = B$ with probability at least $(0.5+\beta)^k$ over the random choice of $(Q,W,B)$.
\end{lem}

\begin{proof}[Proof of Lemma~\ref{le:estimation_k_d}]
	Consider the deterministic function $E:\mathcal{Q} \times \mathcal{W} \to \{0,1\}^k$ defined as follows. For each $(q,w) \in \mathcal{Q} \times \mathcal{W}$, $E(q,w)$ is set to an arbitrary element of the set
	\begin{equation*}
	\argmax_{\hat{B} \in \{0,1\}^k}\Pr[B = \hat{B} | Q=q, W=w].
	\end{equation*}
	
	We now argue that $E(Q,W) = B$ with probability at least $(0.5+\beta)^k$ over the randomness of $(Q,W,B)$, where $\beta$ is a positive constant that only depends on $\alpha$. Since $I(B;W | Q) \geq \alpha \cdot k$, we have that
	\begin{align*}
	H(B | Q,W) &= H(B|Q) - I(B;W|Q)\\ 
	&\le H(B|Q)-\alpha \cdot k\\ 
	&\le H(B) - \alpha \cdot k\\ 
	&\le (1-\alpha)\cdot k,
	\end{align*}
	where the third inequality above uses the fact that conditioning does not increase entropy, and the fourth inequality follows from the fact that $B \in \{0,1\}^k$. By an averaging argument, with probability at least $\alpha/10$ over $(q,w) \sim (Q,W)$, it should be the case that
	\begin{equation}\label{eq:gd_q_w_pairs}
	H(B | Q = q,W =w) \le (1-\alpha/10)\cdot k.
	\end{equation}
	Let $\mathcal{G} \subseteq \mathcal{Q} \times \mathcal{W}$ denote the set of all pairs $(q,w)$ that satisfy Equation~(\ref{eq:gd_q_w_pairs}).  We now fix $(q,w) \in \mathcal{G}$, and consider the min-entropy
	\begin{equation*}
	H_{\min}(B | Q = q,W =w) \triangleq \min_{\hat{B} \in \{0,1\}^k} \log_2\big(\frac{1}{\Pr[B = \hat{B} | Q=q, W=w]}\big).
	\end{equation*}
	Using the fact that min-entropy lower-bounds Shannon entropy and Equation~(\ref{eq:gd_q_w_pairs}), we deduce that there exists $\hat{B} \triangleq \hat{B}(q,w) \in \{0,1\}^k$ such that $\Pr[B = \hat{B} | Q=q, W=w] \geq 2^{-(1-\alpha/10)\cdot k}$. Hence, for any fixed $(q,w) \in \mathcal{G}$, conditioned on $(Q= q, W=w)$, the value $E(q,w)$ is equal to $B$ with probability at least $2^{-(1-\alpha/10)\cdot k}$. Since the probability that $(Q,W) \in \mathcal{G}$ is at least $\alpha/10$, we conclude that $E(Q,W)$ is equal to $B$ with probability at least
	\begin{equation*}
	(\alpha/10)\cdot 2^{-(1-\alpha/10)\cdot k} \geq (0.5+\beta)^k,
	\end{equation*}
	for some constant $\beta$ that only depends on $\alpha$.
\end{proof}

\subsection{Simulation Protocol}\label{subsec:sim_prot}
Recall that the distribution $\nu_{\epsilon}$ over the inputs $((S,X),(T,Y))$ to $F$ was defined as $\nu_{\epsilon} \triangleq \mathcal{D}_{\epsilon} \otimes \mu_{2\epsilon -2\epsilon^2}$. In the following simulation lemma, the error probability will be measured w.r.t. distribution $\nu_{\epsilon}$ whereas the information cost will be measured w.r.t. another distribution $\kappa_{\epsilon}$ over $((S,X),(T,Y))$ inputs, which is defined as $\kappa_{\epsilon} \triangleq \mathcal{D}_{\epsilon} \otimes \mu_{\epsilon}$.

\begin{lem}[Simulation Lemma]\label{le:sim_lem}
	Let $\Pi$ be any deterministic $1$-way protocol computing $F$ with error at most $\theta$ on the distribution $\nu_{\epsilon}$ over $((S,X),(T,Y))$ inputs, and let $M \triangleq M(X,S)$ be the corresponding single message that is sent from Alice to Bob under $\Pi$. Then, we have that
	\begin{equation}\label{eq:sim_lem_bd}
	I_{((S,X),(T,Y)) \sim \kappa_{\epsilon}}\bigg((\langle T^{(i)}, X^{(i)} \rangle)_{i \in [k]} ; M(X,S) ~ | ~ Y,T\bigg) \geq \beta \cdot k
	\end{equation}
	for some constant $\beta > 0$ that only depends on $\theta$.
\end{lem}

We point out that in Lemma~\ref{le:sim_lem}, the error probability is measured w.r.t. the distribution $\nu_{\epsilon}$ while the information cost is measured w.r.t. the distribution $\kappa_{\epsilon}$.
\begin{defn}\label{def:typical}
	A sequence $\hat{T} \triangleq (\hat{T}^{(i)} \subseteq [n])_{i \in [k]}$ of subsets is said to be \emph{typical} if $|\hat{T}^{(i)}| \in [n/3, 2n/3]$ for all $i \in [k]$.
\end{defn}
Recall that the total variation distance between two distributions $\phi$ and $\psi$ defined on the same finite support $\Omega$ is given by $\Delta_{TV}(\phi,\psi) = \max_{A \subseteq \Omega} |\phi(A) - \psi(A)| = 0.5 \cdot \sum_{x \in \Omega} |\phi(x)-\psi(x)|$. We will use the next lemma.

\begin{lem}[Closeness Lemma]\label{le:close_lem}
	For a given sequence $\hat{T} \triangleq (\hat{T}^{(i)} \subseteq [n])_{i \in [k]}$ of subsets, we define the distribution $\mu_{\hat{T},\epsilon}$ as follows. To sample $(X,Y) \sim \mu_{\hat{T},\epsilon}$, we independently sample $U,V \in_R \{0,1\}^k$, $Z \in_R \{0,1\}^{k \cdot n}$, $X$ to be an $\epsilon$-noisy copy of $Z$ conditioned on $(\langle \hat{T}^{(i)}, X^{(i)} \rangle)_{i \in [k]} = U$, and $Y$ to be an $\epsilon$-noisy copy of $Z$ conditioned on $(\langle \hat{T}^{(i)}, Y^{(i)} \rangle)_{i \in [k]} = V$.
	
	Then, for every fixed typical $\hat{T}$, we have that
	\begin{equation*}
	\Delta_{TV}(\mu_{2\epsilon-2\epsilon^2}, \mu_{\hat{T},\epsilon}) \le k \cdot \exp(- \epsilon \cdot n).
	\end{equation*}
\end{lem}

\begin{proof}[Proof of Lemma~\ref{le:close_lem}]
	We denote $\langle \hat{T}, X \rangle \triangleq (\langle \hat{T}^{(i)}, X^{(i)} \rangle)_{i \in [k]}$ and similarly
	$\langle \hat{T}, Y \rangle \triangleq (\langle \hat{T}^{(i)}, Y^{(i)} \rangle)_{i \in [k]}$. Note that both $\langle \hat{T}, X \rangle$ and $\langle \hat{T}, Y \rangle$ are elements of $\{0,1\}^k$. For every fixed $\hat{X}, \hat{Y} \in \{0,1\}^{k \cdot n}$, define $\hat{U} \triangleq \langle \hat{T}, \hat{X} \rangle$ and $\hat{V} \triangleq \langle \hat{T}, \hat{Y} \rangle$. We have that
	\begin{equation*}
	\mu_{2\epsilon-2\epsilon^2}(X = \hat{X} , Y = \hat{Y}) = \frac{1}{2^{k \cdot n}} \cdot (2\epsilon-2\epsilon^2)^{\Delta(\hat{X},\hat{Y})} \cdot (1-2\epsilon+2\epsilon^2)^{k\cdot n - \Delta(\hat{X},\hat{Y})}
	\end{equation*}
	On the other hand, we have that
	\begin{align}
	&\mu_{\hat{T},\epsilon}(X = \hat{X} , Y = \hat{Y})\nonumber\\ 
	&= \mu_{\hat{T},\epsilon}(U = \hat{U}, V =\hat{V} , X = \hat{X} , Y = \hat{Y})\nonumber\\ 
	&= \displaystyle\sum\limits_{\hat{Z} \in \{0,1\}^{k \cdot n}} \mu_{\hat{T},\epsilon}(Z = \hat{Z}, U = \hat{U}, V =\hat{V} , X = \hat{X} , Y = \hat{Y})\nonumber\\ 
	&= \displaystyle\sum\limits_{\hat{Z} \in \{0,1\}^{k \cdot n}} \mu_{\hat{T},\epsilon}(Z = \hat{Z}, U = \hat{U}, V =\hat{V}) \cdot  \mu_{\hat{T},\epsilon}(X = \hat{X} , Y = \hat{Y} | Z = \hat{Z}, U = \hat{U}, V =\hat{V})\nonumber\\ 
	&= \displaystyle\sum\limits_{\hat{Z} \in \{0,1\}^{k \cdot n}} \frac{1}{4^k} \cdot \frac{1}{2^{k \cdot n}} \cdot \mu_{\hat{T},\epsilon}(X = \hat{X} | Z = \hat{Z}, U = \hat{U}) \cdot \mu_{\hat{T},\epsilon}(Y = \hat{Y} | Z = \hat{Z}, V =\hat{V})\label{eq:hatT_X_Y_last}
	\end{align}
	Denote by $N_{\epsilon}(Z)$ the distribution of a random variable that is an $\epsilon$-noisy copy version of $Z$. Then,
	\begin{equation}\label{eq:That_eps_X}
	\mu_{\hat{T},\epsilon}(X = \hat{X} | Z = \hat{Z}, U = \hat{U}) = \frac{\Pr_{X' \sim N_{\epsilon}(Z)}[ X' = \hat{X}]}{\Pr_{X' \sim N_{\epsilon}(Z)}[ \langle \hat{T}, X' \rangle = \hat{U}]},
	\end{equation}
	where
	\begin{align*}
	\Pr_{X' \sim N_{\epsilon}(Z)}[ \langle \hat{T}, X' \rangle = \hat{U}] &= \Pr_{X' \sim N_{\epsilon}(Z)}[ (\langle \hat{T}^{(i)}, X'^{(i)} \rangle)_{i \in [k]}= \hat{U}]\\ 
	&= \displaystyle\prod\limits_{i=1}^k \Pr_{X'^{(i)} \sim N_{\epsilon}(Z^{(i)})}[ \langle \hat{T}^{(i)}, X^{(i)} \rangle = \hat{U}_i]\\ 
	&= \displaystyle\prod\limits_{i=1}^k (\frac{1}{2} \pm \exp(-\epsilon \cdot n))\\ 
	&= \frac{1}{2^k} \cdot \displaystyle\prod\limits_{i=1}^k (1 \pm \exp(-\epsilon \cdot n))\\ 
	&= \frac{1}{2^k} \cdot (1 \pm k\cdot\exp(-\epsilon \cdot n)),
	\end{align*}
	where the third equality above follows from the fact that $\hat{T}$ is typical. Plugging back this last expression in Equation~(\ref{eq:That_eps_X}), we get
	\begin{align}
	\mu_{\hat{T},\epsilon}(X = \hat{X} | Z = \hat{Z}, U = \hat{U}) &= 2^k \cdot \frac{\Pr_{X' \sim N_{\epsilon}(Z)}[ X' = \hat{X}]}{1 \pm k\cdot\exp(-\epsilon \cdot n)}\nonumber\\ 
	&= 2^k \cdot \Pr_{X' \sim N_{\epsilon}(Z)}[ X' = \hat{X}] \cdot (1 \pm k\cdot\exp(-\epsilon \cdot n))\label{eq:hatT_X_last}.
	\end{align}
	Similarly, we have that
	\begin{equation}\label{eq:hatT_Y_last}
	\mu_{\hat{T},\epsilon}(Y = \hat{Y} | Z = \hat{Z}, V =\hat{V}) = 2^k \cdot \Pr_{Y' \sim N_{\epsilon}(Z)}[ Y' = \hat{Y}] \cdot (1 \pm k\cdot\exp(-\epsilon \cdot n)).
	\end{equation}
	Combining Equations~(\ref{eq:hatT_X_Y_last}), (\ref{eq:hatT_X_last}) and (\ref{eq:hatT_Y_last}) yields
	\begin{align}
	\mu_{\hat{T},\epsilon}(X = \hat{X} , Y = \hat{Y}) &= \displaystyle\sum\limits_{\hat{Z} \in \{0,1\}^{k \cdot n}} \frac{1}{2^{k \cdot n}} \cdot  \Pr_{X' \sim N_{\epsilon}(Z)}[ X' = \hat{X}] \cdot \Pr_{Y' \sim N_{\epsilon}(Z)}[ Y' = \hat{Y}] \cdot (1 \pm k\cdot\exp(-\epsilon \cdot n))^2\nonumber\\ 
	&= (1 \pm k\cdot\exp(-\epsilon \cdot n)) \cdot \displaystyle\sum\limits_{\hat{Z} \in \{0,1\}^{k \cdot n}} \frac{1}{2^{k \cdot n}} \cdot  \Pr_{X' \sim N_{\epsilon}(Z)}[ X' = \hat{X}] \cdot \Pr_{Y' \sim N_{\epsilon}(Z)}[ Y' = \hat{Y}]\nonumber\\ 
	&= (1 \pm k\cdot\exp(-\epsilon \cdot n)) \cdot \mu_{2\epsilon-2\epsilon^2}(X = \hat{X} , Y = \hat{Y}),\label{eq:both_dist_proportional}
	\end{align}
	where the last equality above follows from the fact that one way to sample a $(2\epsilon-2\epsilon^2)$-noisy pair $(X,Y)$ is to first sample a uniform-random $Z$, and then independently sample each of $X$ and $Y$ to be an $\epsilon$-noisy copy of $Z$. Using Equation~(\ref{eq:both_dist_proportional}) and the definition of the total-variation distance, we conclude that
	\begin{align*}
	\Delta_{TV}(\mu_{2\epsilon-2\epsilon^2}, \mu_{\hat{T},\epsilon}) &= \frac{1}{2} \cdot \displaystyle\sum\limits_{\hat{X}, \hat{Y} \in \{0,1\}^{k \cdot n}} |\mu_{2\epsilon-2\epsilon^2}(X = \hat{X} , Y = \hat{Y})-\mu_{\hat{T},\epsilon}(X = \hat{X} , Y = \hat{Y})|\\ 
	&\le k\cdot\exp(-\epsilon \cdot n) \displaystyle\sum\limits_{\hat{X}, \hat{Y} \in \{0,1\}^{k \cdot n}} \mu_{2\epsilon-2\epsilon^2}(X = \hat{X} , Y = \hat{Y})\\ 
	&= k\cdot\exp(-\epsilon \cdot n).\qedhere
	\end{align*}
	
\end{proof}

We are now ready to prove Lemma~\ref{le:sim_lem}.
\begin{proof}[Proof of Lemma~\ref{le:sim_lem}]
	Assume for the sake of contradiction that there exists a deterministic $1$-way protocol $\Pi$ that computes $F$ with error at most $\theta$ on $\nu_{\epsilon}$ and that violates Equation~(\ref{eq:sim_lem_bd}). Namely, if we define the \emph{intermediate information cost} of $\Pi$ as
	\begin{equation}\label{eq:iic_def}
	{\IIC}_{\kappa_{\epsilon}}(\Pi ~ | ~ T) \triangleq I_{((S,X),(T,Y)) \sim \kappa_{\epsilon}}\bigg((\langle T^{(i)}, X^{(i)} \rangle)_{i \in [k]} ; M(X,S) ~ | ~ Y,T\bigg),
	\end{equation}
	then we assume that $\IIC_{\kappa_{\epsilon}}(\Pi ~ | ~ T) = o(k)$. We say that a particular value $\hat{T}$ of $T$ is \emph{nice} if it simultaneously satisfies the following three properties:
	\begin{enumerate}
		\item\label{prop:typical} $\hat{T}$ is typical.
		\item\label{prop:error} The conditional error of $\Pi$ w.r.t. to $\nu_{\epsilon}$ conditioned on $T = \hat{T}$ is at most $O(\theta)$.
		\item\label{prop:info_cost} The intermediate information cost condtioned on $T = \hat{T}$ satisfies
		\begin{equation*}
		\IIC_{\kappa_{\epsilon}}(\Pi ~ | ~ T=\hat{T}) = O\bigg(\IIC_{\kappa_{\epsilon}}(\Pi ~ | ~ T)\bigg) = o(k).
		\end{equation*}
	\end{enumerate}
	We now argue that there exists a $\hat{T}$ that is nice. To do so, we show that a random $\hat{T}$ satisfies the above three properties with high probability. First, by Definition~\ref{def:typical}, a Chernoff bound and union bound, a random $\hat{T}$ satisfies property~\ref{prop:typical} with probability at least $1-o(1)$ as long as $k \cdot \exp(-n) = o(1)$. Moreover, by an averaging argument, a random $\hat{T}$ satisfies property~\ref{prop:error} with probability $1-o(1)$. Finally, by an averaging argument and the definition of the conditional mutual information in Equation~(\ref{eq:iic_def}), we get that a random $\hat{T}$ satisfies property~\ref{prop:info_cost} with probability $1-o(1)$. By a union bound, we conclude that a random $\hat{T}$ satisfies all three properties with high probability. Henceforth, we fix such a nice $\hat{T}$ and use it to give a deterministic $1$-way protocol computing the function $\HD_{k}$ w.h.p. over the uniform distribution on $\{0,1\}^{2 \cdot k}$ and with communication $o(k)$ bits. This would contradict the lower bound of \cite{woodruff2007efficient} (i.e., Lemma~\ref{le:woodruff_lb}).
	
	Consider the simulation protocol $\Pi'$ described in Protocol~\ref{alg:sim_prot}. In this protocol, Alice is given as input a binary string $U$ of length $k$ and Bob is given as input a binary string $V$ of length $k$. We will argue that
	\begin{enumerate}[label=(\alph*)]
		\item\label{prop:error_ghd} The output of $\Pi'$ is equal to $\HD_k(U,V)$ with probability $1-O(\theta)$ over the randomness of $(U,V) \sim U_{2k}$ and over the private and shared randomness of $\Pi'$.
		\item\label{prop:info_cost_ghd} The information cost of $\Pi'$ satisfies
		\begin{equation*}\label{eq:info_cost_sim_prot}
		I(U; M'(U) ~ | ~ V, R) = \IIC_{\kappa_{\epsilon}}(\Pi ~ | ~ T=\hat{T}),
		\end{equation*}
		where $M'$ is the single (randomized) message sent from Alice to Bob under $\Pi'$, and $R$ is the public randomness of $\Pi'$.
	\end{enumerate}
	We start by proving property~\ref{prop:error_ghd}. Let $\lambda$ be the probability distribution of the sequence $S$ of subsets that is sampled in Protocol~\ref{alg:sim_prot}. In other words, $S \sim \lambda$ is an $\epsilon$-noisy copy of $\hat{T}$. Then, when $(U,V)$ is drawn uniformly at random, the induced distribution on $(S,X,Y)$ in Protocol~\ref{alg:sim_prot} is $\lambda \otimes \mu_{\hat{T},\epsilon}$. Property~\ref{prop:error} above guaranteed that the error probability of protocol $\Pi$ on pairs $((S,X),(\hat{T},Y))$ such that $(S,X,Y) \sim \lambda \otimes \mu_{\epsilon}$ is at most $O(\theta)$. Using Lemma~\ref{le:close_lem}, the fact that Protocol $\Pi'$ simulates $\Pi$ and the fact that
	\begin{equation*}
	F((S,X),(\hat{T},Y)) = \HD_k\bigg((\langle \hat{T}^{(i)}, X^{(i)} \rangle)_{i \in [k]} , (\langle \hat{T}^{(i)}, Y^{(i)} \rangle)_{i \in [k]} \bigg),
	\end{equation*} we get that the error probability of $\Pi'$ (over the randomness of $(U,V) \sim \mathcal{U}_{2k}$ and over the private and shared randomness) is at most $O(\theta) + O(k \cdot \exp(-\epsilon n))$, which is $O(\theta)$.
	
	We next prove property~\ref{prop:info_cost_ghd}. The information cost of $\Pi'$ is given by
	\begin{align*}
	I(U; M'(U) ~ | ~ V, R) &= I((\langle \hat{T}^{(i)}, X^{(i)} \rangle)_{i \in [k]} ; M(S,X) ~ | ~ V, Z)\\ 
	&= I((\langle \hat{T}^{(i)}, X^{(i)} \rangle)_{i \in [k]} ; M(S,X) ~ | ~ Z)\\ 
	&= I((\langle \hat{T}^{(i)}, X^{(i)} \rangle)_{i \in [k]} ; M(S,X) ~ | ~ Z, T = \hat{T})\\ 
	&= \IIC_{\kappa_{\epsilon}}(\Pi ~ | ~ T=\hat{T})
	\end{align*}
	where the second equality above follows from the fact that $(S,X)$ and $V$ are conditionally independent given $Z$, and the last equality follows from the fact that $(X,Z) \sim \mu_{\epsilon}$ in Protocol~\ref{alg:sim_prot}.
	
	To sum up, the $1$-way protocol $\Pi'$ computes $\HD_k$ with error probability at most $O(\theta)$ and has information cost $o(k)$ bits (by Properties~\ref{prop:info_cost} and \ref{prop:info_cost_ghd} above). By averaging over the shared randomness, we can convert $\Pi'$ into a $1$-way private-coin protocol $\Pi''$ with the same error and information cost guarantees. Note that since $U$ and $V$ are independent, we have that $I(U,V ; M''(U)) = I(U;M''(U) ~ | ~ V)$ where $M''$ is the single message sent from Alice to Bob under $\Pi''$. Applying the generic compression result of \cite{jain2003direct} (i.e., Lemma~\ref{le:jrs_comp_bd_rds}) with error parameter $\zeta = \theta$, we get that there exists a $1$-way deterministic protcol $\Pi'''$ that computes $\HD_k$ with error probability at most $O(\theta)$ over the uniform distribution and with communication cost $o(k)$ bits. This contradicts the lower bound of Woodruff \cite{woodruff2007efficient} (i.e., Lemma~\ref{le:woodruff_lb}).
	\begin{algorithm}
\caption{Simulation Protocol $\Pi'$}
\label{alg:sim_prot}
{\bf Inputs.} Alice is given $U \in \{0,1\}^k$ and Bob is given $V \in \{0,1\}^k$. \\ 
{\bf Parameters.} A fixed sequence $\hat{T} \triangleq (\hat{T}^{(i)} \subseteq [n])_{i \in [k]}$ of subsets and noise parameters $\epsilon, q > 0$. \\
\begin{algorithmic}[1]
\State Alice and Bob use their shared randomness to sample $Z \in_R \{0,1\}^{k \cdot n}$.
\State Alice uses her private randomness to sample $X \in \{0,1\}^{k \cdot n}$ to be an $\epsilon$-noisy copy of $Z$ conditioned on $(\langle \hat{T}^{(i)}, X^{(i)} \rangle)_{i \in [k]} = U$.
\State Bob uses his private randomness to sample $Y \in \{0,1\}^{k \cdot n}$ to be an $\epsilon$-noisy copy of $Z$ conditioned on $(\langle \hat{T}^{(i)}, Y^{(i)} \rangle)_{i \in [k]} = V$.
\State Alice user her private randomness to sample a sequence $S \triangleq (S^{(i)} \subseteq [n])_{i \in [k]}$ of subsets which is set to be a $q$-noisy copy of $\hat{T}$.
\State Alice and Bob simulate the $1$-way deterministic protocol $\Pi$ on inputs $((S,X),(\hat{T},Y))$ and return the resulting output.
\end{algorithmic}
\end{algorithm}

\end{proof}

\subsection{Proof of Lemma~\ref{le:rand_lb_std}}\label{subsec:pf_rand_lb_std}

Assume for the sake of contradiction that there is a deterministic $1$-way protocol $\Pi$ computing $F$ with error at most $\theta$ over the distribution $\nu_{\epsilon}$, and that has communication cost $o(k \cdot \epsilon \cdot n)$ bits. Let $M \triangleq M(X,S)$ be the single message that is sent from Alice to Bob under $\Pi$. By the Simulation Lemma~\ref{le:sim_lem}, we should have that
\begin{equation}\label{eq:iic_lb_sec}
I_{((S,X),(T,Y)) \sim \kappa_{\epsilon}}\bigg((\langle T^{(i)}, X^{(i)} \rangle)_{i \in [k]} ; M(X,S) ~ | ~ Y,T\bigg) \geq \beta \cdot k
\end{equation}
for some constant $\beta$ that only depends on $\theta$.

By Lemma~\ref{le:estimation_k_d} and Equation~(\ref{eq:iic_lb_sec}), there exists a deterministic function $E(Y,T,M(X,S)) \in \{0,1\}^k$ such that $E(Y,T,M(X,S)) = (\langle T^{(i)}, X^{(i)} \rangle)_{i \in [k]}$ with probability at least $(0.5+\gamma)^k$ for some positive constant $\gamma$  that only depends on $\theta$. Hence, by applying the function $E(Y,T,M(X,S))$ to his inputs $(Y,T)$ and to the message $M(X,S)$ that he receives from Alice, Bob can guess the sequence $(\langle T^{(i)}, X^{(i)} \rangle)_{i \in [k]}$ with probability $(0.5+\gamma)^k$. By Corollary~\ref{cor:sdp_disc_comb} -- which combines the strong direct product theorem for discrepancy of \cite{lee2008direct} (i.e., Lemma~\ref{le:strong_dp}) and the base lower bound of \cite{ghazi2015communication} that was proved using the discrepancy method (i.e., Lemma~\ref{le:soda_disc_up_bd}) --  we conclude that the protocol $\Pi$ should have communication cost $\Omega(k \cdot \epsilon \cdot n)$ bits.

\section{Proof of Theorem~\ref{le:isr}}\label{sec:isr_app_pf}

We start by recalling the statement of Theorem~\ref{le:isr}.

\begin{reptheorem}{le:isr}
Let $\rho \in (0,1]$ and $\mu$ be a product distribution. Let $\mathcal{F}$ consist of pairs $(f,g)$ of functions with $\Delta_{\mu}(f,g) \le \delta$, and $\owCC^\mu_{\epsilon}(f), \owCC^\mu_{\epsilon}(g) \le k$. Then, for every positive $\theta$, 
$\owIsrCCU_{\epsilon + 2\delta + \theta}^{\mu}(\mathcal{F}) \leq O_\theta(k/\rho^2)$.
\end{reptheorem}

In order to prove Theorem~\ref{le:isr}, we start by defining a communication problem that will be useful to us (a similar definition is used in \cite{CGMS15}).
\begin{defn}[Gap Inner Product; $\GIP_{c,s}^d$]
Let $-1 \le s < c \le 1$. Let Alice be given a vector $u \in \{\pm 1\}^d$ and Bob be given a vector $v \in \{\pm 1\}^d$. The goal is for Alice and Bob to distinguish the case where $\Ex_{i \in_R [d]}[ u_i v_i ] \geq c$ from the case where $\Ex_{i \in_R [d]}[ u_i v_i ] \le s$.
\end{defn}

We let $\ISR_{\rho}$ denote the source of \emph{imperfect shared randomness} where Alice gets a string $r$ of independent uniform-random bits and Bob gets a string $r'$ of bits obtained by independently flipping each coordinate of $r$ with probability $(1-\rho)/2$. For a function $f$, we denote by $\owIsrCC_{\epsilon, \rho}(f)$ the minimum cost of a protocol that has access to $\ISR_{\rho}$ and that on input $(x,y)$ in the domain of $f$, ouputs $f(x,y)$ with probability at least $1-\epsilon$, where the probability is over the randomness of $\ISR_{\rho}$.

The following theorem --which upper bounds the  communication complexity with $\ISR_{\rho}$ of $\GIP_{c,s}^d$-- was proved by \cite{CGMS15} using a locality-sensitive-hashing based protocol.
\begin{theorem}[\cite{CGMS15}]\label{thm:main_cgms}
Let $\rho \in (0,1]$. Then, $\owIsrCC_{\epsilon, \rho}(\GIP_{c,s}^d) = O((c-s)^{-2} \rho^{-2}\log(1/\epsilon))$ via a protocol where Alice's message depends only her input, her part of the randomness, the values of $\rho$ and $\epsilon$ and the difference $c-s$.
\end{theorem}

Theorem~\ref{thm:main_cgms} can be used in order to estimate the weighted inner product of two vectors up to an arbitrary additive accuracy, and when the weighting is done according to an arbitrary distribution on coordinates that is known to both Alice and Bob.

\begin{lem}\label{le:needed_isr}
Let $t, d \in \mathbb{N}$ and $P$ be a distribution over $[d]$ that is known to both Alice and Bob. Let Alice be given a vector $u \in \{\pm 1\}^d$ and Bob be given $t$ vectors $v^{(1)}, v^{(2)}, \dots, v^{(t)} \in \{\pm1\}^d$. Let $\theta > 0$ and $\rho \in (0,1]$ be given. Then, there exists a $1$-way protocol with communication cost $O(\theta^{-2} \rho^{-2} \log(t/\theta))$ bits such that, with probability $1-\theta$, for every $j \in [t]$, Bob computes $\Ex_{i \sim P} [u_i v^{(j)}_i]$ up to an additive accuracy of $\theta$.
\end{lem}

\begin{proof}
First, note that we can reduce the case of general distributions $P$ on $[d]$ to the case of the uniform distribution on $[d']$ for some integer $d' \in \mathbb{N}$, by having Alice and Bob repeat coordinate $i \in [d]$ of their vectors a number of times proportional to $P(i)$. More precisely, we can assume WLOG that $P(i)$ is a rational number for each $i \in [d]$ (because of the density of the rationals in the reals), and then have Alice and Bob repeat coordinate $i$ a number of times equal to $\ell \cdot P(i)$ where $\ell$ is the least-common multiple of the denominators in $\{P(i): i \in [d]\}$.

Henceforth, we assume that $P$ is the uniform distribution on $[d']$. Alice sends a message of the protocol for $\GIP$ in Theorem~\ref{thm:main_cgms} with parameters $c-s = \theta$ and $\epsilon = \theta^2/t$.
Bob then divides the interval $[-1,+1]$ into $2/\theta$ sub-intervals, each of length $\theta$. Then, he completes the protocol for $\GIP$ in Theorem~\ref{thm:main_cgms} on Alice's message, for each subinterval and for each of his vectors $v^{(1)},v^{(2)}, \dots, v^{(t)}$. For each fixed $j \in [t]$, by a union bound over the $2/\theta$ sub-intrevals, we get that with probability $1-\theta/t$, Bob can deduce the value of $\Ex_{i \sim P} [u_i v^{(j)}_i]$ up to additive accuracy $\theta$. Another union bound over all $t$ vectors of Bob implies that with probablility $1-\theta$, for each $j \in [t]$, he computes the value of $\Ex_{i \sim P} [u_i v^{(j)}_i]$ up to additive accuracy $\theta$.

Moreover, by our setting of $c - s = \theta$ and $\epsilon = \theta^2/t$, we get that the communication cost of the protocol is $O(\theta^{-2} \rho^{-2} \log(t/\theta))$ bits.
\end{proof}

We are now ready to prove Theorem~\ref{le:isr}.

\begin{proof}[Proof of Theorem~\ref{le:isr}]
Let $\rho \in [0,1]$ and $\mu$ be a product distribution. Consider a pair $(f,g) \in \mathcal{F}$. Since $\owCC^\mu_{\epsilon}(g) \le k$, there exist an integer $L \leq 2^k$ and (deterministic) functions $\pi\colon X \to [L]$ and $\{B_i \colon Y \to \zo \}_{i \in [L]}$, such that Alice's message on input $x$ is $\pi(x)$, and Bob's output on message $i$ from Alice and on input $y$ is $B_i(y)$. For every $x \in X$, we define the function $f_x: Y \to \{0,1\}$ as $f_x(y) \triangleq f(x,y)$ for all $y \in Y$. We also denote $\tilde{f}_x(y) \triangleq (-1)^{f_x(y)}$, and similarly $\tilde{B}_i(y) \triangleq (-1)^{B_i(y)}$. The operation of the protocol is given in Protocol~\ref{prot:isr}.


\begin{algorithm}[H]
    \caption{The protocol that handles contextual
      uncertainty with $\ISR_{\rho}$} \label{prot:isr}
    \textbf{Setting:} Let $\mu$ be a product distribution over a message
    space $X\times Y$. Alice and Bob are given functions $f$ and $g$, and inputs
    $x$ and $y$, respectively, where $\Delta_{\mu}(f,g) \le \delta$, $\owCC^\mu_{\epsilon}(f), \owCC^\mu_{\epsilon}(g) \le k$ and
    $(x,y) \sim \mu$.

    \textbf{Protocol:}
    \begin{enumerate}
\item Alice and Bob run the protocol in Lemma~\ref{le:needed_isr} with $t = L$, $d = |Y|$, $P = \mu_Y$, accuracy $\theta/3$, $u = \tilde{f}_x$, $v^{(j)} = \tilde{B}_j$ for every $j \in [L]$.
    \item 
      For every $j \in [L]$, let $\agr_j$ denote Bob's estimate of $\Ex_{y \sim \mu_Y}[\tilde{f}_x(y)\tilde{B}_y(y)]$.
\item Bob determines $j_{\max} \triangleq \mathop{\mathrm{argmax}}_{j\in [L]}\{\agr_j\}$ and outputs $B_{j_{\max}}(y)$ and halts. 
    \end{enumerate}
\end{algorithm}

The same argument as in the analysis of the protocol of \cite{ghazi2015communication} (specialized for product distributions) then implies that the probability that $B_{j_{\max}}(y)$ is not equal to $g(x,y)$ is at most $\epsilon +2\delta + \theta$. By Lemma~\ref{le:needed_isr} and our setting of $t = L \le 2^k$, we get that the communication cost of Protocol~\ref{prot:isr} is $O_{\theta}(k/\rho^2)$ bits, as desired.
\end{proof}

\section{Open Questions}\label{sec:conc}

As mentioned in Notes~\ref{rem:priv} and~\ref{rem:pub} in Section~\ref{sec:intro}, significantly improving the bounds from Theorems~\ref{th:det_sep} and~\ref{th:imp_rand_sep} seems to require fundamentally new constructions, and is a very important question. Moreover, is there an analogue of the protocol in Theorem~\ref{le:isr} for \emph{non-product} distributions?

Another very important and intriguing open question is whether efficient communication under contextual uncertainty is possible in the multi-round setup. Namely, if $k$ is the $r$-round certain communication, can we upper bound the $r$-round uncertain communication by some function of $k$, $I$ and possibly $r$?  Even for $r=2$ and when the uncertain protocol is allowed to use public randomness, no non-trivial protocols are known in this setting. On the other hand, no separations are known for this case (beyond those known for $r=1$) even if the protocols are restricted to be deterministic.

\subparagraph*{Acknowledgements.}\label{sec:ack}

The authors would like to thank Ilan Komargodski, Pravesh Kothari and Mohsen Ghaffari and the anonymous reviewers for very helfpul discussions and pointers.

\bibliographystyle{alpha}
\bibliography{refs}

\appendix

\section{Two-Dimensional Berry-Esseen Theorem for Independent Random Variables}\label{sec:berry_esseen}
In this section, we state a two-dimensional Berry-Esseen theorem for independent (but not necessarily identically distributed) binary random variables that we used in the proofs in Section~\ref{sec:det_sep} (namely, in the proofs of Part~\ref{part:det_dist} of Theorem~\ref{th:det_sep} and of Lemma~\ref{le:dist_sep}). It follows from a known multi-dimensional Berry-Esseen theorem and the argument is very similar to that of \cite{matulef2010testing}, the only exceptions being that in our case the random variables are not necessarily identically distributed, and each of them takes values in $\{-1,0,+1\}$ instead of $\{\pm 1\}$.

\begin{theorem}[Two-dimensional Berry-Esseen]\label{th:2d_berry_ess_bin_not_id}
Consider the linear form $\ell(z) \triangleq k^{-1/2} \cdot \displaystyle\sum\limits_{i \in [k]} z_i$ where $z \in \{-1,0,+1\}^k$. Let $(x,y) \in \{-1,0,+1\}^k \times \{-1,0,+1\}^k$ such that independently for each $i \in [k]$, $(x_i,y_i)$ is a pair of zero-mean random variables with covariance matrix $\Sigma_i$. Let $\Sigma \triangleq k^{-1} \cdot \displaystyle\sum\limits_{i \in [k]} \Sigma_i$ and denote by $\lambda$ the smallest eigenvalue of $\Sigma$. Then, for any intervals $I_1, I_2 \subseteq \mathbb{R}$, it holds that
\begin{equation*}
| \Pr[ (\ell(x),\ell(y)) \in I_1 \times I_2 ] - \Pr[(X,Y) \in I_1 \times I_2] | \le O\bigg(\frac{1}{\lambda^{3/2} \cdot \sqrt{k}}\bigg),
\end{equation*}
where $(X,Y)$ is a pair of zero-mean Gaussians with covariance matrix $\Sigma$.
\end{theorem}


\begin{proof}
The following statement appears as Theorem 16 in \cite{khot2007optimal} and as Corollary 16.3 in \cite{bhattacharya1986normal}.

\begin{theorem}\label{thm:gen_Berry_Esseen}
Let $X_1, \dots, X_k$ be independent random variables taking values in $\mathbb{R}^d$ and satisfying:
\begin{itemize}
\item $\Ex[X_j]$ is the all-zero vector for every $j \in \{1,\dots,k\}$. 
\item $w^{-1} \sum_{j=1}^w \Cov[X_j] = \Sigma$ where $\Cov$ denotes the covariance matrix.
\item $\lambda$ is the smallest eigenvalue of $\Sigma$ and $\Lambda$ is the largest eigenvalue of $\Sigma$.
\item $\rho_3 = k^{-1} \sum_{j=1}^k \Ex[|| X_j ||^3] < \infty$.
\end{itemize}
Let $Q_k$ denote the distribution of $k^{-1/2}(X_1+\dots+X_k)$, let $\Phi_{0,V}$ denote the distribution of the $d$-dimensional Gaussian with mean $0$ and covariance matrix $\Sigma$, and let $\eta = C \lambda^{-3/2} \rho_3 k^{-1/2}$, where $C$ is a certain universal constant. Then, for any Borel set $A$,
$$ | Q_n(A) - \Phi_{0,V}(A)| \le \eta + B(A),$$
where $B(A)$ is the following measure of the boundary of $A$: $B(A) = 2\sup_{y \in \mathbb{R}^d} \Phi_{0,V}((\partial A)^{\eta'} + y)$, $\eta' = \Lambda^{1/2} \eta$ and $(\partial A)^{\eta'}$ denotes the set of points within distance $\eta'$ of the topological boundary of $A$.
\end{theorem}
We now apply Theorem~\ref{thm:gen_Berry_Esseen} with $d=2$ in order to complete the proof of Theorem~\ref{th:2d_berry_ess_bin_not_id}. We are given that for every $i \in \{1,\dots,k\}$, $\Ex[X_i] = \Ex[Y_i] = 0$ and $\Cov[(X_i,Y_i)] = \Sigma_i$. Thus, $k^{-1} \cdot \sum_{j=1}^k \Cov[(X_j,Y_j)] = \sum_{j=1}^k \Sigma_i= \Sigma$. Note that the largest eigenvalue of $\Sigma$ is $\Lambda = O(1)$. Moreover, since each coordinate of our random variables is $\{-1,0,+1\}$-valued, for every $j \in \{1,\dots,k\}$, $\Ex[||X_j||^3] \le 2^{3/2}$. Thus, $\rho_3 \le 2^{3/2}$. Hence, $\eta = O(\lambda^{-3/2} k^{-1/2})$. As in \cite{khot2007optimal,matulef2010testing}, one can check that the topological boundary of any set of the form $I_1 \times I_2$ (where $I_1, I_2 \subseteq \mathbb{R}$ are intervals) is $O(\eta')$. Since $\eta' = \Lambda^{1/2} \eta = O(\eta)$, Theorem~\ref{th:2d_berry_ess_bin_not_id} follows.
\end{proof}

\end{document}